\documentclass[article,reqno]{amsart}
\usepackage{amsmath,amssymb}
\usepackage{latexsym}
\usepackage{epsfig}
\usepackage{cite}
\usepackage{booktabs}
\usepackage{color}
\usepackage{graphicx}
\usepackage{tikz}
\usepackage{pgfplots}
\usepackage{amsthm}
\usepackage{braket}
\numberwithin{equation}{section}

\newtheorem{lemma}{Lemma}[section]

\newcommand{\beq}{\begin{equation}}
\newcommand{\eeq}{\end{equation}}

\title{Quantization and soliton-like solutions for the $\Phi\Psi$-model in a optic fiber} 
\author{F. Belgiorno$^{1,2,4}$, S.L. Cacciatori$^{3,4}$, S. Trevisan$^{3,4}$ \and A. Vigan\`o$^{4,5}$}
\address{\noindent $^1$Dipartimento di Matematica, Politecnico di Milano, Piazza Leonardo 32, IT-20133 Milano, Italy\endgraf
$^2$INdAM-GNFM \endgraf
$^3$Department of Science and High Technology, Universit\`a dell'Insubria, Via Valleggio 11, IT-22100 Como, Italy\endgraf
$^4$INFN sezione di Milano, via Celoria 16, IT-20133 Milano, Italy\endgraf
$^5$Dipartimento di Fisica, Universit\`a degli Studi di Milano, Via Celoria 16, IT-20133 Milano, Italy}

\begin{document}
\maketitle
\begin{abstract}
 In the framework of a mesoscopical model for dielectric media  we provide an analytical description for the electromagnetic field confined in a cylindrical cavity containing a finite dielectric sample. 
This system is apted to simulate the electromagnetic field in a optic fiber, in which two different regions, a vacuum region and a dielectric one, appear. A complete description for the scattering basis is introduced, together with field quantization and 
the two-point function. Furthermore, we also determine soliton-like solutions in the dielectric, 
 propagating in the sample of nonlinear dielectric medium.  
\end{abstract}

\section{Introduction}

Dielectric media in the framework of analogue gravity are an active subject of investigation, with particular 
reference to the Hawking effect in nonlinear dielectrics. See e.g. 
\cite{philbin,cacciatori,belgiorno-prd,rubino2011,petev,finazzi2012,finazzi2013,hopfield-hawking,jacquet,linder,
hopfield-kerr,haw-book,master}. As to experiments  
with dielectric media and their debate one may refer to \cite{faccio,rubino2011,unshu,rebut,finazzi2013,petev} and also to 
the (uncontroversial) experiment in a optic fiber reported in \cite{drori}. 
In general, the problem is quite difficult, because of dispersive effects 
associated with condensed matter systems. Notwithstanding, a framework can be provided where, 
in the limit of weak dispersive effects, in a precise mathematical sense, one is able to find 
how the Hawking effect manifests itself when the system is affected by the presence of 
horizon(s) (mathematically, turning point(s)) \cite{master}. Mostly, calculations are carried out for a dielectric 
medium filling all the space. Furthermore, in order to avoid technical difficulties arising mainly 
because of the gauge field nature of the electromagnetic field, which arise naturally in the 
Hopfield model (see \cite{hopfield-hawking,hopfield-scripta,exact}),   we have introduced a simplified model, 
called the $\Phi\Psi$-model, where the original fields of the Hopfield model are replaced by two scalar 
fields: $\Phi$ in place of the electromagnetic field, and $\Psi$ in place of the polarization \cite{hopfield-hawking}.  
An exact quantization for the  fully relativistically covariant version of the model have been provided in \cite{phi-psi}. 
We have also taken into account the case of dielectric medium 
filling only an half-space \cite{vigano}. We have verified that, in the latter case, spectral boundary conditions 
are required, because of the peculiar role played by the polarization field.

Herein, we extend our analysis by taking into account a cylindrical geometry, where the dielectric field fills only 
a finite cylindrical region of length $2L$ and radius $R$. The remaining region of radius $R$ is 
filled by vacuum. This simplified setting can be still interesting because the vacuum regions 
could be also replaced by regions containing dielectrics with different refractive index.

In the first part of our analysis, we discuss in details the problem of the boundary conditions to 
be imposed on the fields, and a complete scattering basis for the problem is introduced through separation 
of variables allowed by our peculiar geometrical setting. We also provide the full propagator for the 
model at hand.

In the second part, we take into account a fully nonlinear dielectric, where the nonlinearity 
is simulated by introducing a term proportional to the fourth power of the polarization field $\Psi$. 
Our aim is to show that solitonic solutions exist, representing a dielectric perturbation travelling 
with constant velocity in the direction of the cylindrical fiber axis. We can show that, 
by taking into account `homogeneous' solutions which do not vary in the radial direction and also in the azimuthal 
one, soliton-like solutions exist, with different characteristics depending on suitable parameters. 
A linearization around the solitonic solution is the natural set-up for studying the analog Hawking effect also 
in the present case. The present analysis is in preparation of the 
analysis of the Hawking effect in the geometrical setting described above.

\section{The relativistic Kerr-$\Phi\Psi$ model in a cylindrical fibre}\label{sec:classical}
Let us consider the electromagnetic field in a cylindrical cavity, along the $z$ direction, where an Hopfield dielectric is at rest in the lab, filling the region $C_{\chi}=\{ (t,x,y,z)| -L\leq z\leq L,\ x^2+y^2\leq R^2 \}$. 
In general, we consider inertial frames which are boosted in the $z$
direction with respect to the Lab frame. See also figure \ref{fig:geom}. 

\begin{figure}[hbtp!]
\includegraphics[scale=0.55]{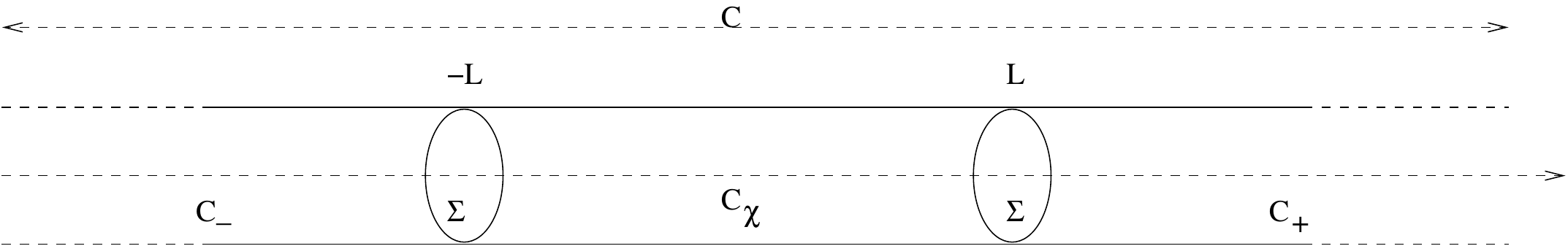} 
\caption{The geometry for the problem at hand is displayed. The region $C_\chi$ contains the dielectric medium, the other two cylindrical 
regions $C_\pm$ with the same radius $R$  are void. The inner boundaries indicated in the text as $\Sigma_{\pm L}$ are for simplicity 
both indicated with $\Sigma$. Again, for simplicity, we have not indicated the other boundaries we take into consideration in the main text.}
\label{fig:geom}
\end{figure}

If $\pmb n$ is the four-vector with covariant components $\underline n=(0,0,0,1)$ in the Lab frame, then, in an arbitrary inertial frame centered in $\pmb 0$ (the origin of the Lab frame) the confining cylindrical region
is $C=\{ \pmb x \in M^{1,3}_{\pmb 0}|  -(\pmb x- (\pmb x \cdot \pmb v)\pmb v+ (\pmb x \cdot \pmb n)\pmb n)^2\leq R^2 \}$, where $\underline v=(1,0,0,0)$ in the Lab frame, and the dielectric region is 
$C_\chi=\{ \pmb x \in C_\chi | -L\leq \pmb x\cdot \pmb n \leq L\}$.

The $\phi-\psi$ model is thus described by the action principle as follows:
\beq
S[\phi,\psi]=\frac 12 \int_C\partial_\mu \phi \partial^\mu \phi\ d^4x +\int_{C_\chi} \biggl(\frac 12 v^\mu\partial_\mu \psi v^\nu \partial_\nu \psi-\frac {\omega_0^2}2 \psi^2 -g\phi v^\mu \partial_\mu \psi\biggr) d^4x.
\eeq
This is because we require for $\psi$ to vanish outside $C_\chi$ by definition.
\subsection{Equations of motion} \label{sec1}
There are several interesting discussions regarding the deduction of the equations of motion. Nevertheless we will follow a simple deduction, by using local variations and then by choosing boundary conditions. 
With ``local variations'' we mean the following. Fix a point $p$ internal to $C_\chi$ (in the topological sense) or internal to $C\backslash C_\chi$. Thus, there is at least an open set $U(p)$ 
such that is completely contained in $C_\chi$ (or in $C\backslash C_\chi$). A local variation is a variation of the fields with support in such a $U(p)$. Using local variations we get for the equations of motion
\begin{align}
\Box \phi &=0, \label{phi-0}\\
\psi &=0 ,
\end{align}
outside $C_\chi$, and
\begin{align}
\Box \phi +gv^\mu \partial_\mu \psi&=0, \label{phi-D}\\
(v^\mu\partial_\mu)^2 \psi +\omega_0^2\psi -gv^\mu\partial_\mu \phi&=0,
\end{align}
inside $C_\chi$. Now, we are left with the choice of the boundary conditions in order to completely define the theory. The boundary consists in $\partial C$, which includes the conditions at infinity and $C_\chi \cap \partial C$, and $\partial C_\chi$ 
that adds $\partial C_\chi -C_\chi \cap \partial C=: \Sigma_{-L}\cup\Sigma_L$ in obvious notations. In order to choose such conditions, let us start by considering (global) variations in $\psi$. The support of $\psi$ is compact, so we choose to work with
variations which are $\mathcal C^\infty(C_\chi)$ with support in $C_\chi$. In particular, we do not require for them to be continuous on $\Sigma_{-L}$ and $\Sigma_L$. Nevertheless, since $\pmb v$ is orthogonal to $\pmb N$, where the latter is the suitably oriented 
normal field to $\partial C$ and $\partial C_\chi$, there are no boundary terms in $\delta S$ under variations of $\psi$, and then we are not required to choose any particular condition on $\psi$ (apart from requiring that it must be at least 
$\mathcal C^2(C_\chi)$).

A little bit more involved is the (global) variation in $\phi$. In this case we have to tackle the variation
\beq
\delta \frac 12 \int_C\partial_\mu \phi \partial^\mu \phi\ d^4x= \int_C\partial_\mu \delta\phi \partial^\mu \phi\ d^4x= \int_C\partial_\mu (\delta\phi \partial^\mu \phi)\ d^4x- \int_C \delta\phi \partial_\mu \partial^\mu \phi\ d^4x.
\eeq
Requiring for $\phi$ to be continuous in $C$ implies that $\delta\phi$ must be continuous. However, requiring also the continuity of $\partial_\mu\phi$ looks too much restrictive in general. In order to manipulate the divergence in the last expression,
let us notice that if we do not require for $\partial_\mu \phi$ to be continuous on $\Sigma_0$ and $\Sigma_L$, then we cannot apply the divergence theorem directly but we need to separate $C$ into three regions as 
$C=C_\chi \cup C_- \cup C_+$, where $C_-\cup C_+:=C-C_\chi$ (with obvious notation). This way
\beq
\int_C\partial_\mu (\delta\phi \partial^\mu \phi)\ d^4x=\int_{C_-}\partial_\mu (\delta\phi \partial^\mu \phi)\ d^4x+\int_{C_+}\partial_\mu (\delta\phi \partial^\mu \phi)\ d^4x+\int_{C_\chi}\partial_\mu (\delta\phi \partial^\mu \phi)\ d^4x,
\eeq
and $\delta\phi \partial^\mu \phi$ is continuous and, indeed, smooth in each of the three regions. So we can apply the divergence theorem to each of the three regions. The result is that,  
if we assume that $n^\mu \partial_\mu \phi$ is continuous in an open neighbourhood of $\Sigma_A$, $A=-L,L$, 
then
\beq
\int_C\partial_\mu (\delta\phi \partial^\mu \phi)\ d^4x=\int_{\partial C} N^\mu \partial_\mu \phi\ \delta\phi\ d^3\sigma.
\eeq
Since we do not mean to fix the value of $\phi$ on the boundary, we can get rid of the boundary term by imposing the Neuman condition $N^\mu \partial_\mu \phi|_{\partial C}=0$. More precisely this condition is clear on the cylindrical boundary
but it should also include a condition at $z\to \pm \infty$. There, the above Neuman condition looks not suitable if we want to allow for sources or, say, fluxes. In this sense it seems that at infinity some other condition could be better, but it is not
a case of interest here.

On the boundaries $\Sigma_A$ of the dielectric we are left with the condition of continuity of the normal derivative of $\phi$. Let us investigate a little bit more at this condition by looking at the equations (\ref{phi-0}) and (\ref{phi-D}).
They show that $\Box \phi$ is not continuous on $\Sigma_A$. Now, since $\psi$ vanishes outside $C_\chi$, we can write both these equations as
\beq
\partial^\mu (\partial_\mu \phi +gv_\mu \psi)=0.
\eeq
Since $\psi$ is discontinuous, $\partial^\mu \psi$ is expected to produce $\delta$-function contributions supported on $\Sigma_A$. However, $\pmb v\cdot \pmb n =0$, so that $v^\mu \partial_\mu \psi$ is discontinuous but does not contains
$\delta$ contributions. Thus, the same happens for $\Box \phi$. Now, since $C$ is contractible we can add to $\pmb n$ and $\pmb v$ two other vectors on $C$, $\pmb e_i$ ($i=1,2$) in order to get a complete constant orthonormal frame. 
Thus, we can write
\beq
0=\partial_\mu (v^\mu  \partial_{\pmb v}\phi -n^\mu  \partial_{\pmb n}\phi -e_1^\mu  \partial_{\pmb e_1}\phi -e_2^\mu  \partial_{\pmb e_2}\phi)=
\partial_{\pmb v} \partial_{\pmb v}\phi -\partial_{\pmb n}  \partial_{\pmb n}\phi -\partial_{\pmb e_1}  \partial_{\pmb e_1}\phi -\partial_{\pmb e_2}  \partial_{\pmb e_2}\phi.
\eeq
Now, $\partial_{\pmb e_1}\phi $, $\partial_{\pmb e_2}\phi $ and $\partial_{\pmb v}\phi $ are discontinuous, but no $\delta$ contribution arises in further deriving, since $\partial_{\pmb e_1} $, $\partial_{\pmb e_2} $ and $\partial_{\pmb v}$
derive in directions orthogonal to $\pmb n$ (and, so, tangent to the separating hypersurface). The remaining term $\partial_{\pmb n}\phi$ is continuous and the further derivative $\partial_{\pmb n}$ introduce at most new discontinuities. 
This shows that no further conditions are necessary for having consistent equations: all the condition we have to require inside $C$ are the continuity of $\phi$ and $\partial_{\pmb n}\phi$ everywhere, with the last vanishing on $\partial C$.

\subsection{General solution.}
Let us work in a frame with four-velocity
\beq
\underline v=\gamma(\nu)(1,0,0,\nu),
\eeq
so that
\beq
\underline n=\gamma(\nu)(\nu,0,0,1).
\eeq
\subsubsection{Outside the dielectric}
The dielectric region is defined by $-\nu t -L/\gamma(\nu) \leq z \leq -\nu t+L/\gamma(\nu).$\\
If we choose cylindrical coordinates, outside the dielectric the equations of motion are simply $\psi=0$ and 
\beq
\partial^2_t \phi -\partial^2_z \phi -\partial^2_\rho \phi -\frac 1\rho \partial_\rho \phi -\frac 1{\rho^2} \partial^2_\theta \phi=0.
\eeq
Separating the variables as 
\beq
\phi(t,\rho,z,\theta)=K \phi_T(t) \phi_R(\rho) \phi_Z(z) \phi_\Theta(\theta),
\eeq
with $K$ a constant, we find that
\beq
\phi_T(t)=e^{-ik_0 t}, \quad \phi_Z(z)=e^{ik_z z}, \quad  \phi_\Theta(\theta)=e^{im\theta},
\eeq
with $m\in \mathbb Z$, and
\beq
\phi''_R+\frac 1\rho \phi'_R+\left( k_\rho^2 -\frac {m^2}{\rho^2} \right) \phi=0,
\eeq
where $k_\rho$ must satisfy $k_0^2-k_z^2-k_\rho^2=0$. The only solutions continuous in $\rho=0$ are $\phi_R(\rho)=J_m(k_\rho \rho)$, where $J_m$ are the usual Bessel functions. The boundary condition on $\partial C$ reduces to
\beq
J'_m(k_\rho R)=0,
\eeq
so that at the end we have
\beq
\phi_R(\rho)=J_{|m|,s}(\rho)=J_{|m|}\biggl(z_{ms}\frac \rho{R}\biggr), \quad |m|,s\in \mathbb N,
\eeq
where $z_{ms}$ is the $s$-th positive zero of $J'_m$. The corresponding dispersion relations are $k_0^2=k_z^2+z_{ms}^2/R^2$, which can be codified in
\beq
K=c_{m,s}(k_0,k_z) \delta \biggl(k_0^2-k_z^2 -\frac {z_{ms}^2}{R^2}\biggr).
\eeq
In conclusion, we can write the general solution outside the dielectric in the form
\begin{align}
&\phi=\sum_{s\in \mathbb N}\sum_{m\in\mathbb Z}\int_{-\infty}^\infty \frac {dk_z}{4\pi k^0}\left(c_{m,s}(k_z) e^{-ik^0t+ik_zz+im\theta} +c^*_{m,s}(k_z) e^{ik^0t-ik_zz-im\theta}\right) J_{|m|}\biggl(z_{ms}\frac \rho{R}\biggr),\\
&\psi=0,
\end{align}
where
\beq
k^0(k_z,m,s)=\sqrt {k_z^2+\frac{z_{ms}^2}{R^2}}.
\eeq

\subsubsection{Inside the dielectric}
Inside the dielectric, the equations take the form ($v^0=\gamma(\nu)$ and $v=\nu\gamma(\nu)$)
\begin{align}
\partial^2_t \phi -\partial^2_z \phi -\partial^2_\rho \phi -\frac 1\rho \partial_\rho \phi -\frac 1{\rho^2} \partial^2_\theta \phi+gv^0\partial_t \psi+gv\partial_z \psi&=0,\\
(v^0\partial_t+v\partial_z)^2\psi+\omega_0^2\psi-gv^0\partial_t\phi-gv\partial_z\phi&=0.
\end{align}
By means of the separation ansatz
\begin{align}
\phi(t,\rho,z,\theta)&=\tilde \phi  \phi_T(t) \phi_R(\rho) \phi_Z(z) \phi_\Theta(\theta),\\
\psi(t,\rho,z,\theta)&=\tilde \psi  \psi_T(t) \psi_R(\rho) \psi_Z(z) \psi_\Theta(\theta),
\end{align}
we find that
\begin{align}
&\psi_T(t)=\phi_T(t)=e^{-ik_0 t}, \quad \psi_Z(z)=\phi_Z(z)=e^{ik_z z}, \quad  \psi_\Theta(\theta)=\phi_\Theta(\theta)=e^{im\theta},\\
&\psi_R(\rho)=\phi_R(\rho)=J_{|m|,s}(\rho)=J_{|m|}\biggl(z_{ms}\frac \rho{R}\biggr), \quad s\in \mathbb N,\ m\in \mathbb Z,
\end{align}
and $\tilde \phi$, $\tilde \psi$ must satisfy the algebraic system
\beq
\begin{pmatrix}
-k_0^2+k_z^2+k_\rho^2 & ig\omega \\
-ig\omega & \omega_0^2-\omega^2
\end{pmatrix}
\begin{pmatrix}
 \tilde \phi \\ \tilde \psi
\end{pmatrix}
=
\begin{pmatrix}
0\\0
\end{pmatrix},
\eeq
where we have put $k_\rho=z_{ms}/R$ and
\beq
\omega=v^\mu k_\mu=k^0 v^0-vk_z.
\eeq
This has nontrivial solutions if the determinant of the matrix vanishes, which means 
\beq
DR:=k_0^2-k_z^2-k_\rho^2+\frac {g^2\omega^2}{\omega_0^2-\omega^2}=0.
\eeq
Notice that this is an implicit equation in $k_0=k^0$, since $\omega$ is a function of $k_0$.
With this condition, the solution of the algebraic system takes the form
\begin{align}
\tilde \phi&=b(k_0,k_z,m,s)\delta(DR),\\
\tilde \psi&=\frac {k_0^2-k_z^2-k_\rho^2}{ig\omega} b(k_0,k_z,m,s)\delta(DR).
\end{align}
If we set 
\beq
DR':=\partial_{k_0} DR=2k^0 \left(1+ \frac {g^2\omega_0^2 v^0}{(\omega_0^2-\omega^2)^2}\right) ,
\eeq
then we can write the general solution inside the dielectric as
\begin{align*}
&\phi=\sum_{a=1}^2\sum_{s\in \mathbb N}\sum_{m\in\mathbb Z}\int_{-\infty}^\infty \frac {dk_z}{2\pi DR'_{(a)}}\left(b_{(a)ms}(k_z) e^{-ik_{(a)}^0t+ik_zz+im\theta} +c.c.\right) J_{|m|}\biggl(z_{ms}\frac \rho{R}\biggr),\cr
&\psi=\sum_{a=1}^2\sum_{s\in \mathbb N}\sum_{m\in\mathbb Z}\int_{-\infty}^\infty \frac {dk_z}{2\pi DR'_{(a)}}\left(b_{(a)ms}(k_z) e^{-ik_{(a)}^0t+ik_zz+im\theta} -c.c.\right) 
\frac {-ig\omega_{(a)}}{\omega^2_{(a)}-\omega^2_0}J_{|m|}\biggl(z_{ms}\frac \rho{R}\biggr),
\end{align*}
where $(a)$ indicates the branch like in \cite{phi-psi,exact}.

\subsubsection{Gluing conditions}
At this point we must impose the continuity of $\phi$ and $\partial_{\pmb n}\phi$.
From now on, we will work in the lab frame where $\underline v\equiv (1,0,0,0)$ and $\omega=k^0$.
In this case, the normal direction is $z$ and the gluing conditions are
\begin{gather}
\phi(-L^-) = \phi(-L^+), \quad \phi(L^-) = \phi(L^+), \notag \\
\label{cond}
\partial_z\phi(-L^-) = \partial_z\phi(-L^+), \quad \partial_z\phi(L^-) = \partial_z\phi(L^+),
\end{gather}
where we defined for short
\beq
\phi(z^\pm):=\lim_{\epsilon\to 0} \phi(t,\rho,z\pm |\epsilon|,\theta).
\eeq
It is convenient to define a scattering basis, defined by replacing the Fourier modes with solutions of the form (omitting the angular and radial parts, at time $t=0$)
\beq
\begin{split}
\phi_k(z) =&
\left(e^{ik z} +R_{ksm} e^{-ikz})\chi_{(-\infty,-L)}(z)\right) + \\
& + \left(M_{ksm} e^{iq_{s,m}(k)z}+N_{ksm} e^{-iq_{sm}(k)z}\right)\chi_{[-L,L]}(z)+ T_{ksm} e^{ikz} \chi_{(L,\infty)},
\end{split}
\eeq
going from the left to the right, and the analogous left moving modes.
Here $R_{ksm}$, $T_{ksm}$, $M_{ksm}$, and $N_{ksm}$ are reflection and transmission coefficients. Moreover, we choose the measure to be $dk_z/(2\pi) 2k^0 $ everywhere so that inside the dielectric 
$b_{(a)ms}$ and $c_{ms}$ reabsorb the normalization factor $1+ \frac {g^2\omega_0^2 v^0}{(\omega_0^2-\omega^2)^2}$ from $DR'$.

Writing the conditions~\eqref{cond} explicitly, we obtain the following algebraic system
\begin{equation}
\begin{cases}
e^{-ikL} + R_{ksm} e^{ikL} = M_{ksm} e^{-iqL} + N_{ksm} e^{iqL} \\
k e^{-ikL} - R_{ksm} k e^{ikL} = M_{ksm} q e^{-iqL} - N_{ksm} q e^{iqL} \\
T_{ksm} e^{ikL} = M_{ksm} e^{iqL} + N_{ksm} e^{-iqL} \\
T_{ksm} k e^{ikL} = M_{ksm} q e^{iqL} - N_{ksm} q e^{-iqL}
\end{cases} 
\end{equation}
that has solution
\begin{align}
R_{ksm}=&\frac {2i(k^2-q^2_{sm}(k)) \sin (2q_{sm}(k)L) e^{2i(q_{sm}(k)-k)L}}{(k-q_{sm}(k))^2e^{4iq_{sm}(k)}-(k+q_{sm}(k))^2},\\
M_{ksm}=& -\frac{2k (k+q_{sm}(k)) e^{i(q_{sm}(k)-k)L}}{(k-q_{sm}(k))^2e^{4iq_{sm}(k)}-(k+q_{sm}(k))^2}, \\
N_{ksm}=& \frac{2k (k-q_{sm}(k)) e^{i(3q_{ms}(k)-k)L}}{(k-q_{sm}(k))^2e^{4iq_{sm}(k)}-(k+q_{sm}(k))^2}, \\
T_{ksm}=&-\frac {4kq_{sm}(k) e^{2iq_{sm}(k)L}e^{-2ikL}}{(k-q_{sm}(k))^2e^{4iq_{sm}(k)}-(k+q_{sm}(k))^2}.
\end{align}

\subsection{The scattering basis}
The positive energy scattering basis consists in the dielectric modes and the gap modes. In the lab frame, the dielectric modes are (we write the $\phi$-component only):
\begin{align}
\begin{split}
\phi^R_{D,ksm}(t,\rho,z,\theta)&=\kappa_{ksm} e^{-i\omega_{ksm}t} e^{im\theta} J_{|m|}\biggl(z_{ms}\frac \rho{R}\biggr) \Bigl[ \Bigl(e^{ik z} +R_{ksm} e^{-ikz}\Bigr)\chi_{(-\infty,-L)}(z) \\
& + \Bigl(M_{ksm} e^{iq_{s,m}(k)z}+N_{ksm} e^{-iq_{sm}(k)z}\Bigr)\chi_{[-L,L]}(z)+ T_{ksm} e^{ikz} \chi_{(L,\infty)} \Bigr],
\end{split}
\\
\phi^L_{D,ksm}(t,\rho,z,\theta)& = \phi^R_{D,ksm}(t,\rho,-z,\theta),
\end{align}
where 
\begin{align}
\omega_{ksm}&=\sqrt{k^2+k_\rho^2},\\
q_{sm}(k)^2+k_\rho^2&=(k^2+k_\rho^2)\frac {\omega_0^2+g^2-k^2-k_\rho^2}{\omega_0^2-k^2-k_\rho^2},\\
m\in \mathbb Z, & \qquad\ s\in \mathbb N, \qquad\ k>0,
\end{align}
and the coefficients are
\begin{align}
R_{ksm}=&\frac {2i(k^2-q^2_{sm}(k)) \sin (2q_{sm}(k)L) e^{2i(q_{sm}(k)-k)L}}{(k-q_{sm}(k))^2e^{4iq_{sm}(k)L}-(k+q_{sm}(k))^2},\\
M_{ksm}=& -\frac{2k (k+q_{sm}(k)) e^{i(q_{sm}(k)-k)L}}{(k-q_{sm}(k))^2e^{4iq_{sm}(k)L}-(k+q_{sm}(k))^2}, \\
N_{ksm}=& \frac{2k (k-q_{sm}(k)) e^{i(3q_{ms}(k)-k)L}}{(k-q_{sm}(k))^2e^{4iq_{sm}(k)L}-(k+q_{sm}(k))^2}, \\
T_{ksm}=&-\frac {4kq_{sm}(k) e^{2iq_{sm}(k)L}e^{-2ikL}}{(k-q_{sm}(k))^2e^{4iq_{sm}(k)L}-(k+q_{sm}(k))^2}.
\end{align}
Notice that 
\beq
|T_{ksm}|^2+|R_{ksm}|^2=1, 
\eeq
so that there is no trapping in the dielectric.
The dielectric modes are defined in the range for $\omega_{k,s,m}$ not in the gap $[\omega_0,\bar\omega]$, and do not form a complete basis for all possible initial conditions.

In order to get a complete basis we have to add the gap modes
\begin{align}
\phi^R_{G,ksm}(t,\rho,z,\theta)&=\tilde \kappa_{ksm} e^{-i\omega_{ksm}t} e^{im\theta} J_{|m|}\biggl(z_{ms}\frac \rho{R}\biggr)\sin((k+L)z) \chi_{(-\infty,-L]}(z), \\
\phi^L_{G,ksm}(t,\rho,z,\theta)&=\tilde \kappa_{ksm} e^{-i\omega_{ksm}t} e^{im\theta} J_{|m|}\biggl(z_{ms}\frac \rho{R}\biggr)\sin((k-L)z) \chi_{[L,\infty)},
\end{align}
defined for 
\beq
\omega_0^2 \leq k^2+k_\rho^2 \leq \bar\omega^2.
\eeq
These represent modes that are totally reflected by the dielectric.

The normalisation constants $\kappa_{ksm}$ and $\tilde \kappa_{ksm}$ can be computed by using the results in appendix~\ref{orthogonality}. If we choose
\beq
\kappa_{ksm}=\frac { \tilde \kappa_{ksm}}2=\frac 1{ \left( 1-\frac {m^2}{z^2_{ms}} \right)^{\frac 12} \sqrt \pi R J_{|m|}(z_{ms})},
\eeq
then the scattering solution are orthonormalised (with measure $dk/{2\pi (2k^0)}$). 

\section{Quantization}
Here we determine the scalar product and invert the expressions to compute the amplitudes fields $c$, $b$ in terms of the fields and their conjugate momenta, and impose the equal time canonical commutation relations (ETCCR).

We write the full field as a superposition of the component fields $\phi^R_D$ etc, so:
\beq
\label{fullfield}
\begin{split}
\phi(t,z,\rho,\theta) & = \sum_{s,m} \int_{R_{m,s}} \frac {dk}{4\pi k^0} \bigg\{ a^R_{D,ksm} \phi^R_{D,ksm}(t,z,\rho,\theta) + a^L_{D,ksm} \phi^L_{D,ksm}(t,z,\rho,\theta) + \\
&\quad + d^R_{G,ksm} \phi^R_{G,ksm}(t,z,\rho,\theta) + d^L_{G,ksm} \phi^L_{G,ksm}(t,z,\rho,\theta) + \text{ h.c.} \biggl\} ,
\end{split}
\eeq
where $R_{m,s}$ is the range of $k$ satisfying the right spectral conditions for any given $s,m$. In the particular case when $s=m=0$ one has that the integration is on $[0,\omega_0] \cup [\bar \omega,\infty)$ for $D$ modes and $[\omega_0,\bar \omega]$
for the $G$ modes (see \cite{phi-psi}).

Now, our purpose is to determine the commutation relations between the operators $a$'s and $d$'s: to do this, we note that
the scalar product defined in appendix~\ref{orthogonality} can also be written as
\begin{equation}
(f|\tilde f) = \Braket{\Psi, \tilde{\Psi}} := \frac{i}{2} \int d^3 \, \Psi \Omega \tilde{\Psi} ,
\end{equation}
where
\begin{equation}
\Psi =
\begin{pmatrix}
\phi \\ \psi \\ \pi_\phi \\ \pi_\psi
\end{pmatrix}
\end{equation}
and
\begin{equation}
\Omega =
\begin{pmatrix}
0_{2x2} & 1_{2x2} \\
-1_{2x2} & 0_{2x2}
\end{pmatrix} .
\end{equation}
Given the orthogonality relations, we can write the coefficients as the scalar product
\begin{equation}
a^R_{D,ksm} =  \Braket{\Psi^R_{D,ksm},\Psi_{ksm}} ,
\end{equation}
and similarly fo the other coefficients.

With some algebra, we can evaluate the products $a^R_{D,ksm} {a^R_{D,k's'm'}}^\dagger$ and ${a^R_{D,k's'm'}}^\dagger a^R_{D,ksm}$,
and imposing the ETCCR $[\phi(x),\pi_\phi(y)] = i\delta^{(3)}(x-y)$, we get that the commutator is equal to the scalar product
\begin{equation}
[a^R_{D,ksm}, {a^R_{D,k's'm'}}^\dagger] = \Braket{\Psi^R_{D,ksm},\Psi^R_{D,k's'm'}} .
\end{equation}
Therefore we can easily evaluate the commutation relations, using the orthogonality relations calculated in appendix~\ref{orthogonality}.
The relevant commutators are
\begin{align}
\Bigl[a^R_{D,ksm}, {a^R_{D,k's'm'}}^\dagger\Bigr] & = 4\pi \omega_{ksm} \delta_{mm'} \delta_{ss'} \delta(k'-k) , \\
\Bigl[a^L_{D,ksm}, {a^L_{D,k's'm'}}^\dagger\Bigr] & =  4\pi \omega_{ksm} \delta_{mm'} \delta_{ss'} \delta(k'-k) , \\
\Bigl[d^R_{G,ksm}, {d^R_{G,k's'm'}}^\dagger\Bigr] & =  4\pi \omega_{ksm} \delta_{mm'} \delta_{ss'} \delta(k'-k) , \\
\Bigl[d^L_{G,ksm}, {d^L_{G,k's'm'}}^\dagger\Bigr] & =  4\pi \omega_{ksm} \delta_{mm'} \delta_{ss'} \delta(k'-k) ,
\end{align}
and all other vanish.

\section{The two-point function of the Kerr-$\Phi\Psi$ model}


The full quantum theory, in absence of the Kerr nonlinear term, is fully defined by the free propagator.
Let us compute the two-point function of the free theory
\begin{align}
iG_{\psi\psi}^0(x,x') = \left. \braket{\psi(x)\psi(x')}\right|_{\lambda=0} \,.
\end{align}
In general, its explicit expression will depend on where we choose the points $x$ and $x'$. We will not consider the gap modes. In fact, the right and left gap modes can be included by noting that they are equivalent to the $\phi^R$ and $\phi^L$ 
modes respectively, with reflection coefficient $R=1$ and all other coefficients vanishing, so that, if necessary, the corresponding contributions can be deduced by taking the limit $R\to1$.

We start by discussing the case $x,x' \in C_\chi$, which is the most important one for computing Feynman diagrams in perturbation theory, when the nonlinearity is included. For simplicity, we restrict to the case $\rho=\rho'$, $\theta=\theta'$ in the following 
computation, but at the end we will give the general result:
\beq
\begin{split}\label{eq_Propag1}
iG_{\psi\psi}^0(x,x')= \theta(t-t') \sum_{s,m} \int_0^{\infty} \frac{dk}{4\pi \omega_{ksm}} \left[ \psi_{ksm}^R(x){\psi_{ksm}^R(x')}^* + \psi_{ksm}^L(x){\psi_{ksm}^L(x')}^* \right] + (x \leftrightarrow x')  \\
= \frac i {2\pi} \int d\omega \sum_{s,m} \int_0^{\infty} \frac{dk}{4\pi \omega_{ksm}}\frac{e^{-i(\omega+\omega_{ksm})(t-t')}}{\omega + i\epsilon} \Big[ \left( M_{ksm}e^{iq(k)z} + N_{ksm}e^{-iq(k)z} \right)\left( M_{ksm}^*e^{-iq(k)z'} + N_{ksm}^*e^{iq(k)z'} \right)  \\
+ \left( M_{ksm}e^{-iq(k)z} + N_{ksm}e^{iq(k)z} \right)\left( M_{ksm}^*e^{iq(k)z'} + N_{ksm}^*e^{-iq(k)z'} \right)\Big] \, \frac{g^2\omega_{ksm}^2\kappa_{sm}^2|J_m(\rho)|^2}{(\omega_{ksm}^2 - \omega_0^2)^2} \\
-\frac i {2\pi} \int d\omega \sum_{s,m} \int_0^{\infty} \frac{dk}{4\pi \omega_{ksm}}\frac{e^{i(\omega+\omega_{ksm})(t-t')}}{\omega + i\epsilon} \Big[ \left( M_{ksm}e^{iq(k)z'} + N_{ksm}e^{-iq(k)z'} \right)\left( M_{ksm}^*e^{-iq(k)z} + N_{ksm}^*e^{iq(k)z} \right) \\
+ \left( M_{ksm}e^{-iq(k)z'} + N_{ksm}e^{iq(k)z'} \right)\left( M_{ksm}^*e^{iq(k)z} + N_{ksm}^*e^{-iq(k)z} \right)\Big] \, \frac{g^2\omega_{ksm}^2\kappa_{sm}^2|J_m(\rho)|^2}{(\omega_{ksm}^2 - \omega_0^2)^2}  \\
= \frac i {2\pi} \int d\tilde\omega \sum_{s,m} \int_0^{\infty} \frac{dk}{4\pi \omega_{ksm}} \Big[ 2\left(|M_{ksm}|^2 + |N_{ksm}|^2\right)\cos(q(k)(z-z')) + 2\text{Re}(M_{ksm}N_{ksm}^*)\cos(q(k)(z+z')) \Big] \\
\times \frac{g^2\omega_{ksm}^2\kappa_{sm}^2|
J_m(\rho)|^2}{(\omega_{kms}^2 - \omega_0^2)^2} e^{-i\tilde\omega(t-t')} \left( \frac{1}{\tilde\omega -\omega_{kms} + i\epsilon} - \frac{1}{\tilde\omega +\omega_{kms} - i\epsilon} \right)\,.
\end{split}
\eeq
In the second step we used the integral representation $\theta(\tau) = \frac i {2\pi} \int d\omega \frac{e^{-i\omega\tau}}{\omega + i\epsilon}$ for the Heaviside $\theta$-function, and in the third step we have performed the change of variables 
$\tilde\omega= \omega + \omega_{ksm}$ in the first integral, and $\tilde\omega=- \omega - \omega_{ksm}$ in the second one. Notice that the two-point function (\ref{eq_Propag1}) can be written as the sum of two functions
\beq
G_{\psi\psi}^0(x,x') = G_1(t-t',\rho,\theta,z-z') + G_2(t-t',\rho,\theta,z+z') \,. 
\eeq
This makes evident the breaking of translation invariance along the $z$ axis. While $G_1$ is translationally invariant, $G_2$ can be interpreted as depending on the reflections at the boundaries of the dielectric region.

Let us focus first on the $G_1$ part. By performing a change of variable $k = k_a(q)$, where the subscript $a=\pm$ denotes the two branches of the dispersion relation, we find
\beq
\begin{split}
 G_1(t-t',\rho,\theta,z-z') &=  \sum_{s,m,a}\int \frac {d\omega}{2\pi} \int_{0}^{\infty} \frac{dq}{2\pi DR'_a(q)} e^{-i\tilde\omega(t-t')} \frac{g^2\omega_{asm}^2(q)\kappa_{sm}^2|J_m(\rho)|^2}{(\omega_{asm}^2(q) 
 - \omega_0^2)^2} \frac{2\omega_{asm}(q)}{\omega^2 - \omega_{asm}^2(q) +i\epsilon }  \\
&\quad \times \left(|M_{asm}(q)|^2 + |N_{asm}(q)|^2\right)\left( e^{iq(z-z')} + e^{-iq(z-z')} \right) \,,
\end{split}
\eeq
where
\begin{align}
\omega_{asm}(q) &:= \omega_{k_a(q)sm}, \\
M_{asm}(q) &:= M_{k_a(q)sm}, \\
DR_a'(q) &= 2\omega_{asm}(q)\left(1 + \frac{g^2\omega_0^2}{(\omega_{asm}^2(q) - \omega_0^2)^2}\right).
\end{align}
Notice that $\omega_{asm}(q)\equiv \omega_a \left(\sqrt{q^2 + k_\rho^2}\right)$, where $\omega_a(\cdot)$ denotes the two solutions of the dispersion relation as in \cite{phi-psi}. 
After defining $k_a(-q):=-k_a(q)$, we can rewrite the integral in $q$ over the whole $(-\infty,+\infty)$ range. Performing the sum over $a$ explicitly and extending and noting that $M_{asm}(-q)=M_{asm}^*(q)$, we obtain the final expression
\beq
\begin{split}\label{eq_Propag3}
&G_1(t-t',\rho,\theta,z-z') =  \sum_{s,m} \int \frac {d\omega}{2\pi} \frac {dq}{2\pi} \kappa_{sm}^2|J_m(\rho)|^2 e^{-i\omega(t-t') + iq(z-z')}   \frac{\omega^2 - q^2 - k_\rho^2}{(\omega^2 - q^2 - k_\rho^2)(\omega^2 - \omega_0^2) - g^2\omega^2}   \\ &
\times\frac{1} {(\omega^2 - q^2 - k_\rho^2)\,
(\omega_{+sm}^2(q) - \omega_{-sm}^2(q))}  [(\omega^2 - \omega_{-sm}^2)(\omega_{+sm}^2 - q^2 - k_\rho^2)(|M_{+sm}(q)|^2+|N_{+sm}(q)|^2) \\
&- (\omega^2 - \omega_{+sm}^2)(\omega_{-sm}^2 - q^2 - k_\rho^2)(|M_{-sm}(q)|^2+|N_{-sm}(q)|^2)]    \\
&=:  \sum_{s,m} \int \frac {d\omega}{2\pi} \frac {dq}{2\pi} \,\kappa_{sm}^2|J_m(\rho)|^2 \,e^{-i\omega(t-t') + iq(z-z')}  D(\omega,q,s,m)\, \ell_1(\omega,q,s,m) \,,
\end{split}
\eeq
where $D(\omega,q,s,m)$ is the factor in the first line, which equals the free propagator $G_{\psi\psi}^0$ computed in \cite{phi-psi} for the case of the infinite dielectric medium. Also we defined $\ell_1(\omega,q,s,m)$,  
which includes the corrections due to the reflections and transmissions appearing in the finite dielectric. Notice that in the case $|M_q|^2+|N_q|^2=1$ we have $\ell_1(\omega,q) = 1$.
From the expression (\ref{eq_Propag3}) we can read the Fourier transform of the $G_1$ part of the propagator. It has the same poles as the ones in the infinite dielectric case, corresponding to the Sellmeier dispersion relation.

It can be of interest to understand the asymptotic behaviour of the factor $\ell_1$ for large momenta. For the $|M_{asm}(q)|^2+|N_{asm}(q)|^2$ factor, we have
\beq
\label{eq_MNfactor}
|M_{asm}(q)|^2+|N_{asm}(q)|^2 = \frac{k_a^2(q)\left(k_a^2(q) + q^2\right)}{\left(k_a^2(q) + q^2\right)^2 \sin^2(2qL) + 4q^2 k_a^2(q) \cos^2(2qL)} \,.
\eeq
In the case $a=+$, we have $k_a(q) \sim q$ for $q \gg \omega_0$. Therefore, in the limit $q\rightarrow \infty$ the whole factor (\ref{eq_MNfactor}) tends to $1/2$.
In the case $a=-$, instead, $\lim_{q\rightarrow\infty} k_a(q) = \omega_0$, and equation~\eqref{eq_MNfactor} has a point dependent limit. Indeed, we can take a succession ${q_n}$ such that $\sin(2q_n L) = 0$, in such a way that for $n\rightarrow\infty$ 
equation (\ref{eq_MNfactor}) tends to $1/4$. For all other successions such that $\sin(2q_n L) = C \neq 0$, we have
$$ |M_{-sm}(q)|^2+|N_{-sm}(q)|^2 \sim \frac{\omega_0^2}{q^2 \sin^2C} \rightarrow 0 \,.$$
Therefore, we see that the Fourier transform of the propagator~\eqref{eq_Propag3} is $\sim q^{-2}$, with the exception of small neighbourhoods of the points $q_n=n\pi/(2L)$, where sharp peaks of height $1/4$ appear.


For the $G_2$ part of the propagator, with similar manipulations, we obtain
\beq
\begin{split}
&G_2(t-t',\rho,\theta,z+z') =  \sum_{s,m} \int \frac {d\omega}{2\pi} \frac {dq}{2\pi} \kappa_{sm}^2|J_m(\rho)|^2  e^{-i\omega(t-t') + iq(z+z')} \frac{\omega^2 - q^2 - k_\rho^2}{(\omega^2 - q^2 - k_\rho^2)(\omega^2 - \omega_0^2) - g^2\omega^2} \\
& \times \frac{(\omega^2 - \omega_{-sm}^2)(\omega_{+sm}^2 - q^2 - k_\rho^2)\text{Re}(M_{+sm}(q)N_{+sm}(q)^*) - (\omega^2 - \omega_{+sm}^2)(\omega_{-sm}^2 - q^2 - k_\rho^2)\text{Re}(M_{-sm}(q)N_{-sm}(q)^*)} {(\omega^2 - q^2 - k_\rho^2)\,
(\omega_{+sm}^2(q) - \omega_{-sm}^2(q))} \\
&=:  \sum_{s,m} \int \frac {d\omega}{2\pi} \frac {dq}{2\pi} \,\kappa_{ms}^2|J_m(\rho)|^2 \, e^{-i\omega(t-t') + iq(z+z')}  D(\omega,q,s,m)\, \ell_2(\omega,q,s,m)
\end{split}
\eeq
The factors $M_{asm}(q)N_{asm}(q)^*$ have a very similar behaviour as the factors $|M_{asm}(q)|^2+|N_{asm}(q)|^2$ studied before, so also the function $G_2$ is vanishing for large $q$.

The final expression of the propagator, in the general case $\rho'\neq\rho$, $\theta'\neq\theta$, is therefore:
\beq
\begin{split}
G_{\psi\psi}^0(t-t',\rho,\rho',\theta-\theta',z,z')\,\chi_{[-L,L]}(z)\,\chi_{[-L,L]}(z')=  \sum_{s,m} \int \frac{d\omega}{2\pi} \frac{dk}{2\pi}  e^{-i\omega(t-t')} J_{|m|,s}(\rho)J_{|m|,s}(\rho')e^{im(\theta-\theta')} \\
\times \kappa_{sm}^2D(\omega,q,s,m) \left[ \ell_1(\omega,q,s,m)
e^{ iq(z-z')} + \ell_2(\omega,q,s,m) e^{ iq(z+z')}  \right] .
\end{split}
\eeq
The remaining components of the propagator, $G_{\phi\phi}$ and $G_{\phi\psi}$, can be derived in a very similar way. We express the result as a matrix propagator
\begin{align}
\begin{pmatrix}
G_{\phi\phi}^0 & G_{\phi\psi}^0 \\ G_{\psi\phi}^0 & G_{\psi\psi}^0
\end{pmatrix}\,\chi_{[-L,L]}(z)\,\chi_{[-L,L]}(z') = \mathcal{G}_1(\tau,\rho,\rho',\Theta,\xi_-) + \mathcal{G}_2(\tau,\rho,\rho',\Theta,\xi_+) \,,
\end{align}
where $\tau=t-t'$, $\Theta = \theta-\theta'$, $\xi_\mp=z\mp z'$. 
The matrices $\mathcal{G}_{a}$, $a=1,2$, result as the Fourier transform in $\tau$ and $\xi_\mp$ of
\begin{align}
\tilde{\mathcal{G}}_{a}(\omega,\rho,\rho',\Theta,k) = \sum_{s,m} \kappa_{sm}^2J_{|m|,s}(\rho)J_{|m|,s}(\rho')e^{im\Theta} \frac{1}{(\omega^2 - k^2 - k_\rho^2)(\omega^2 - \omega_0^2) - g^2\omega^2} \mathcal M_{a}\,,
\end{align}
The matrix $\mathcal{M}_1$ has the following components (here we omit the obvious dependences on $m$, $s$ and $q$):
\begin{align}
\mathcal{M}_1^{\phi\phi} &= \frac{(\omega^2 - \omega_{-}^2)(\omega_{+}^2 - \omega_0^2)(|M_{+}|^2+|N_{+}|^2) - (\omega^2 - \omega_{+}^2)(\omega_{-}^2 - \omega_0^2)(|M_{-}|^2+ |N_{-}|^2)} {(\omega_{+}^2 - \omega_{-}^2)} ,\\
\mathcal{M}_1^{\phi\psi} &= \left(\mathcal{M}_1^{\psi\phi} \right)^* =  -ig\omega\, \frac{(\omega^2 - \omega_{-}^2)(|M_{+}|^2+|N_{+}|^2) - (\omega^2 - \omega_{+}^2)(|M_{-}|^2+ |N_{-}|^2)} {(\omega_{+}^2 - \omega_{-}^2)} , \\
\mathcal{M}_1^{\psi\psi} &= \frac{(\omega^2 - \omega_{-}^2)(\omega_{+}^2 - q^2 - k_\rho^2)(|M_{+}|^2+|N_{+}|^2) - (\omega^2 - \omega_{+}^2)(\omega_{-}^2 - q^2 - k_\rho^2)(|M_{-}|^2+ |N_{-}|^2)} {(\omega_{+}^2 - \omega_{-}^2)} .
\end{align}
The matrix $\mathcal M_2$ is obtained from $\mathcal{M}_1$ by replacing the factors $(|M_a|^2+|N_a|^2)$ by $\text{Re}(M_aN_a^*)$.

The calculations for the case $z<-L$, $z'>L$ and $z,z'<-L$  are very similar, and actually simpler, to the previous case. The $\phi$--$\phi$ propagator in this case is the only non vanishing one, and it is given by
\begin{align}
\begin{split}
G_{\phi\phi}^0&(t-t',\rho,\theta,z,z') \,\chi_{[-\infty,-L]}(z)\,\chi_{[L,+\infty]}(z') =  \\ &\sum_{s,m} \int \frac{d\omega}{2\pi} \frac{dk}{2\pi}  e^{-i\omega(t-t')} \kappa_{sm}^2 J_{|m|,s}(\rho)J_{|m|,s}(\rho')e^{im(\theta-\theta')}  \frac{1}{\omega^2 - \omega_{ksm}^2} T_{ksm}^*e^{ik(z-z')}  .
\end{split} \\
\begin{split}
	G_{\phi\phi}^0&(t-t',\rho,\theta,z,z') \,\chi_{[-\infty,-L]}(z)\,\chi_{[-\infty,-L]}(z') = \\  &\sum_{s,m} \int \frac{d\omega}{2\pi} \frac{dk}{2\pi}  e^{-i\omega(t-t')} \kappa_{sm}^2 J_{|m|,s}(\rho)J_{|m|,s}(\rho')e^{im(\theta-\theta')} 
	 \frac{1}{\omega^2 - \omega_{ksm}^2} 
	\Big[ e^{ik(z-z')} + R_{ksm}^*e^{ik(z+z')} \Big] .
\end{split}
\end{align}
In the second case again we find a dependence on $z+z'$ which is due to reflections; in the first case this does not happen due to the fact that $R_{kms} T_{kms}^*$ is purely imaginary, so that $R_{kms} T_{kms}^*+ R_{kms}^* T_{kms}=0$. We also notice that for $T=1$ and $R=0$ we find the usual Feynman propagator for a scalar field. 

The propagators in the cases $z<-L$, $-L<z'<L$ and  $-L<z<L$, $z'>L$ have less trivial dependence on $z$ and $z'$. The non vanishing components in this cases are  
\begin{align}
\begin{split}
G_{\phi\phi}^0&(t-t',\rho,\theta,z,z') \,\chi_{[-\infty,-L]}(z)\,\chi_{[-L,L]}(z') = \\ &\sum_{s,m} \int \frac{d\omega}{2\pi} \frac{dk}{2\pi}  e^{-i\omega(t-t')}\frac{\kappa_{sm}^2J_{|m|,s}(\rho)J_{|m|,s}(\rho')e^{im(\theta-\theta')}}{\omega^2 - \omega_{ksm}^2} \Big[  M_{ksm}^*  e^{ikz-iq(k)z'} + N_{ksm}^* e^{ikz+iq(k)z'} \Big] ,
\end{split}
\\
\begin{split}
G_{\phi\psi}^0&(t-t',\rho,\theta,z,z') \,\chi_{[-\infty,-L]}(z)\,\chi_{[-L,L]}(z') = \\ &\sum_{s,m} \int \frac{d\omega}{2\pi} \frac{dk}{2\pi}  e^{-i\omega(t-t')}\frac{\kappa_{sm}^2J_{|m|,s}(\rho)J_{|m|,s}(\rho')e^{im(\theta-\theta')}}{\omega^2 - \omega_{ksm}^2} \,
\frac{-ig\omega}{\omega_{ksm} ^2 - \omega_0^2} \Big[  M_{ksm}^*  e^{ikz-iq(k)z'} + N_{ksm}^* e^{ikz+iq(k)z'} \Big] ,
\end{split}
\\
\begin{split}
G_{\psi\phi}^0&(t-t',\rho,\theta,z,z') \,\chi_{[-\infty,-L]}(z)\,\chi_{[-L,L]}(z') = \\ &\sum_{s,m} \int \frac{d\omega}{2\pi} \frac{dk}{2\pi}  e^{-i\omega(t-t')}\frac{\kappa_{sm}^2J_{|m|,s}(\rho)J_{|m|,s}(\rho')e^{im(\theta-\theta')}}{\omega^2 - \omega_{ksm}^2} \,  \frac{ig\omega}{\omega_{ksm}^2 - \omega_0^2}
\Big[  M_{ksm}^*  e^{ikz-iq(k)z'} + N_{ksm}^* e^{ikz+iq(k)z'} \Big]  ,
\end{split}
\end{align}
where the factor in square brackets is the same for all three components. These expressions are obtained by making use of the following equalities, that are easily checked:
\begin{align}
M_{kms} = R_{kms} N_{kms}^* + T_{kms} M_{kms}^* \,, \\
N_{kms} = R_{kms} M_{kms}^* + T_{kms} N_{kms}^* \,.
\end{align}

\section{Solitonic solutions}

In this section we introduce a nonlinearity in the model, with the aim 
of describing the perturbation of refractive index propagating in the nonlinear dielectric medium when 
a strong laser pulse is shot into the dielectric and the Kerr effect is stimulated. The propagating perturbation 
breaks the homogeneity of the dielectric sample described in the previous section. Still, the solutions 
for the homogeneous case represent a good asymptotic scattering basis for the full nonlinear problem in the 
linearisation of the theory around the dielectric perturbation represented by the solitonic solutions we are going 
to describe.

With this aim, we add a fourth order term 
in the polarization field $\psi$, as in \cite{vigano}.  A fourth order term in nonlinear optical
media appears also e.g.~in  \cite{drummond-book}, where a fourth order term in the displacement
field can be introduced in the case of an optical fiber. We have discussed a fourth order term in the 
polarization field for the Hopfield model in \cite{hopfield-kerr}, where we have shown that our solitonic 
solutions can be associated with the Kerr effect in a proper way. It is also to be remarked that our approach does not
represent the standard way to approach the Kerr effect (see also \cite{ablowitz,boyd}). and that our solitonic solutions are more constrained that the usual solutions of the nonlinear Schr\"odinger equation studied in  \cite{ablowitz}. 
See also the discussion in \cite{universe}. 

Our non-linear theory has the following action
\begin{equation}
S[\phi,\psi] = \int_C \frac{1}{2} \partial_\mu \phi \partial^\mu \phi \, d^4x +
\int_{C_\chi} \biggl[ \frac{1}{2} (v^\mu \partial_\mu \psi)^2 - \frac{{\omega_0}^2}{2} \psi^2 -
g \phi v^\mu \partial_\mu \psi - \frac{\lambda}{4!} \psi^4 \biggr] \, d^4x ;
\end{equation}
the equations of motion are
\begin{align}
\square\phi + gv^\mu\partial_\mu\psi &= 0 , \\
(v^\mu\partial_\mu)^2\psi + \omega_0^2\psi - gv^\mu\partial_\mu\phi + \frac{\lambda}{3!} \psi^3 &= 0 .
\end{align}
In lab $v^\mu=(1,0,0,0)$, so the equations become
\begin{align}
\square\phi + g\dot{\psi} &= 0 , \\
\ddot{\psi} + \omega_0^2\psi - g\dot{\phi} + \frac{\lambda}{3!} \psi^3 &= 0 .
\end{align}
We are not looking for the general solution of the system above, still we are interested in finding out analytical solutions.  
We attempt the ansatz
\begin{gather}
\phi(t,z,\rho,\theta) = f(z-Vt) Y(\rho,\theta) , \\
\psi(t,z,\rho,\theta) = h(z-Vt) Y(\rho,\theta) ;
\end{gather}
the radial and angular parts can be separated if $Y=const$~\footnote{This would be equivalent to a calculation 
involving only $s$-waves in presence of spherical symmetry.} to obtain the equations
\begin{align}
(1-V^2)f'' + gVh' &= 0 , \\
V^2 h'' + gVf' + \omega_0^2h + \frac{Y^2\lambda}{3!}h^3 &= 0 ,
\end{align}
where the prime stands for the derivative with respect to the argument $z-Vt$.

Integrating the first equation and inserting in the second, after a new integration we obtain
\begin{equation}
\label{quartic}
\frac{Vh'}{i\sqrt{\frac{Y^2\lambda}{12}h^4 + (\omega_0^2-g^2V^2\gamma^2)h^2 + 2g\kappa Vh - 2\chi}} = 1 ,
\end{equation}
where $\kappa$ and $\chi$ are constants of integration, and $\gamma=(1-V^2)^{-\frac 12}$. We will also use the notation
\begin{align}
 v:=\gamma V.
\end{align}
It is worth to mention here that we accept solution having finite energy. It is immediate to see that the energy density inside the dielectric is
\beq
\mathcal E=\frac 12 \dot \phi^2+\frac 12 \nabla \phi \cdot \nabla \phi+\frac 12 \dot \psi^2+\frac 12 \omega_0^2\psi^2 +\frac {\lambda}{4!}\psi^4
=Y^2(\chi + (gV\gamma^2 h(z-Vt)-\kappa)^2-(1-V^2)\kappa).
\eeq
Thus, the energy is finite if $h$ has no poles. Notice that, in particular, $h$ has to be limited, which implies that the quartic radicand in \eqref{quartic} must have real roots. So, the constants must be constrained in order to ensure this condition.

\subsection{$\kappa=\chi=0$}
In this case, the solution inside the dielectric is the solitonic solution obtained in \cite{vigano}. 
With $a=\sqrt{ \frac{12}{\lambda Y^2} (g^2v^2 - \omega_0^2) }$ and
$b=\frac{1}{v} \sqrt{g^2v^2 - \omega_0^2 }$, we find (for $-L \leq z \leq L$)
\begin{align}
h & = \frac{a}{\cosh(b\gamma(z-Vt))} , \\
f & = \frac{2agv}{b} \arctan \biggl[ \tanh \biggl(\frac{b}{2}\gamma(z-Vt)\biggr) \biggr] ,
\end{align}
while the solution in vacuum is a superposition of a progressive and a regressive wave,
whose form is determined by the continuity of $\phi$ and $\partial_z\phi$, as discussed in section~\ref{sec1}:
\begin{align}
\partial_z f=
\begin{cases}
 \frac{1+V}{2} \frac{ag\gamma v}{\cosh[b\gamma(L(1-V)-V(z-t))]} + \frac{1-V}{2} \frac{ag\gamma v}{\cosh[b\gamma(L(1+V)+V(z+t))]}, & z\le -L\\
 \frac{1+V}{2} \frac{ag\gamma}{\cosh[b\gamma(L(1-V)+V(z-t))]} + \frac{1-V}{2} \frac{ag\gamma}{\cosh[b\gamma(L(1+V)-V(z+t))]}, & z\ge L
\end{cases} .
\end{align}
Note that the solution exists only if $v>\omega_0/g$, that is
\begin{align}
 V^2>\frac {\omega_0^2}{g^2+\omega_0^2}.\label{condition}
\end{align}

\subsection{$\kappa\neq0$, $\chi\neq0$}
We look for a M\"obius transformation 
\begin{equation}
h=\frac {as+b}{cs+d} ,
\end{equation}
which maps~\eqref{quartic} into the form
\begin{equation}
\frac iV=\frac {s'}{\sqrt {4s^3-g_2s-g_3}}. \label{equation}
\end{equation}
The assumption that the original quartic equation has at least one real root ensures that such a transformation exists with real coefficients $a,b,c,d$ that can be chosen to satisfy $ad-bc=\pm 1$, see Appendix \ref{appendice-B}.
Equation \eqref{equation} has general solution
\begin{equation}
s(x)=\wp(g_2,g_3;i(x-x_0)/V)=-\wp(g_2,-g_3;(x-x_0)/V),
\end{equation}
where $g_2$, $g_3$ are defined in Appendix \ref{appendice-B} and $\wp$ is the Weierstrass elliptic function defined by 
\begin{align}
\wp(z)=\frac 1{z^2} +\sum_{(n,m)\in \mathbb Z^2-\{0,0\}} \left( \frac 1{(z+n\omega_1+m\omega_2)^2}-\frac 1{(n\omega_1+m\omega_2)^2} \right)
\end{align}
with $\omega_1$, $\omega_2$ the two periods satisfying $\tau:=\omega_2/\omega_1\notin \mathbb R$,
and $x_0$ is an integration constant. Indeed, we can be more precise and notice that there are two distinct situations, when all three roots of $4z^3-g_2z-g_3$ are distinct. In our case, $g_2$ and $g_3$ are real and so we may have three real roots 
$e_3<e_2<e_1$, or one real root $e_2$ and two complex roots $e_1, e_3$, with $e_1=\bar e_3$. Let us shortly discuss the two cases.

\subsubsection{3 REAL ROOTS} $e_3<e_2<e_1$. The periods of $\wp(g_2,-g_3,z)$ are
\begin{align}
 \omega_1&= 2\int_{e_1}^\infty \frac {dz}{\sqrt {4z^3-g_2z+g_3}}\ \in \mathbb R_{>}, \\
 \omega_2&= 2i\int_{-\infty}^{e_3} \frac {dz}{\sqrt {-4z^3+g_2z-g_3}}\ \in i\mathbb R_{>},
\end{align}
and we get two distinct solutions for $h$. The first one is
\begin{equation}
h(z-Vt)=\frac {a \wp(g_2,-g_3;t-(z-z_0)/V)-b}{c\wp(g_2,-g_3;t-(z-z_0)/V)-d}.\label{weierstrass-1}
\end{equation}
In this case $\wp$ assumes all values in $[e_1,\infty)$. Since we are interested in solutions with $h$ of class $C^2$ everywhere inside the dielectric, we must discard solutions such that the denominator above vanishes somewhere. This happens only if the
condition
\beq
 ce_1 -d>0 \label{condizione-1}
\eeq
is satisfied. Analysing this condition in general depends on several details. A partial analysis can be found in appendix~\ref{appendice-B}. This solutions represent trains of pulses having period $\omega_1$, and moving with constant velocity as we can see from figure~\ref{w1}.
\begin{figure}
\centering
\includegraphics[scale=0.5]{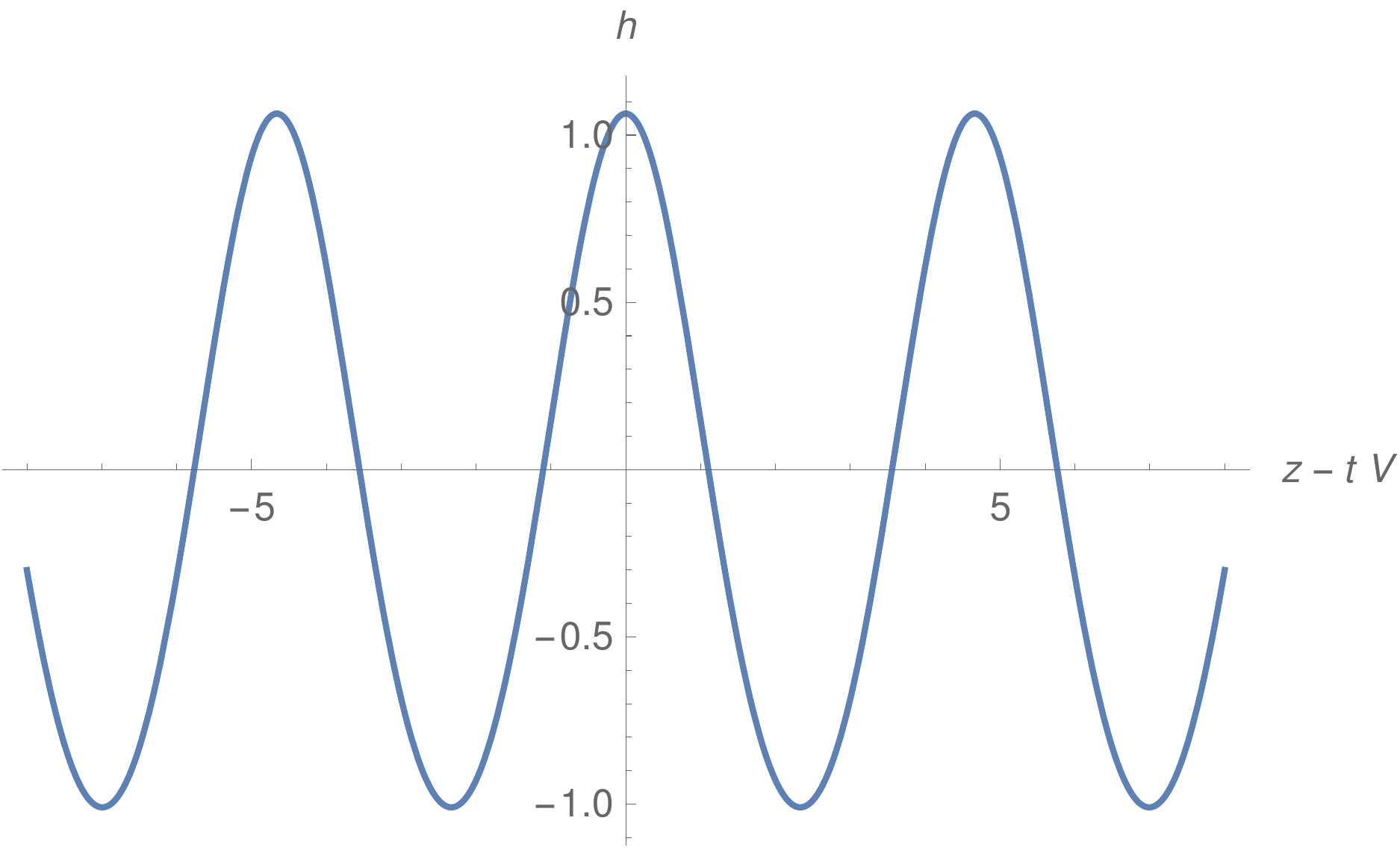}
\caption{Plot of the solution~\eqref{weierstrass-1} corresponding to the three real root case. The parameters have the following values:
$Y=\omega_0=g=1$, $\lambda=V=0.5$, $\chi=2$, $\kappa=3$, $z_0=0$.}
\label{w1}
\end{figure}

The second solution is 
\begin{equation}
\label{weierstrass-2}
h(z-Vt)=\frac {a \wp(g_2,-g_3;t-(z-z_0)/V+\frac {\omega_2}2)-b}{c\wp(g_2,-g_3;t-(z-z_0)/V+\frac {\omega_2}2)-d}.
\end{equation}
In this case $\wp$ oscillates in the interval $[e_3,e_2]$, which is acceptable if the condition 
\beq
 \frac{d}{c} \notin [e_3,e_2]
\eeq
is satisfied. Again, we get a train of pulses which is shifted along the horizontal axis, as shown in figure~\ref{w2}.

\subsubsection{1 REAL ROOT}  $e_2\in\mathbb R$. In this case we have two complex conjugate periods $\omega_1=\omega$ and $\omega_2=\bar \omega$ with
\begin{align}
 \omega=\int_{e_2}^\infty \frac {dz}{\sqrt {4z^3-g_2z+g_3}}+i \int_{-\infty}^{e_2} \frac {dz}{\sqrt {-4z^3+g_2z-g_3}}.
\end{align}
In this case there is only one kind of solutions, having the form (\ref{weierstrass-1}) with condition (\ref{condizione-1}). This represents a train of pulses having period $\omega+\bar\omega$, like for example in figure~\ref{w3}.
\begin{figure}
\centering
\includegraphics[scale=0.5]{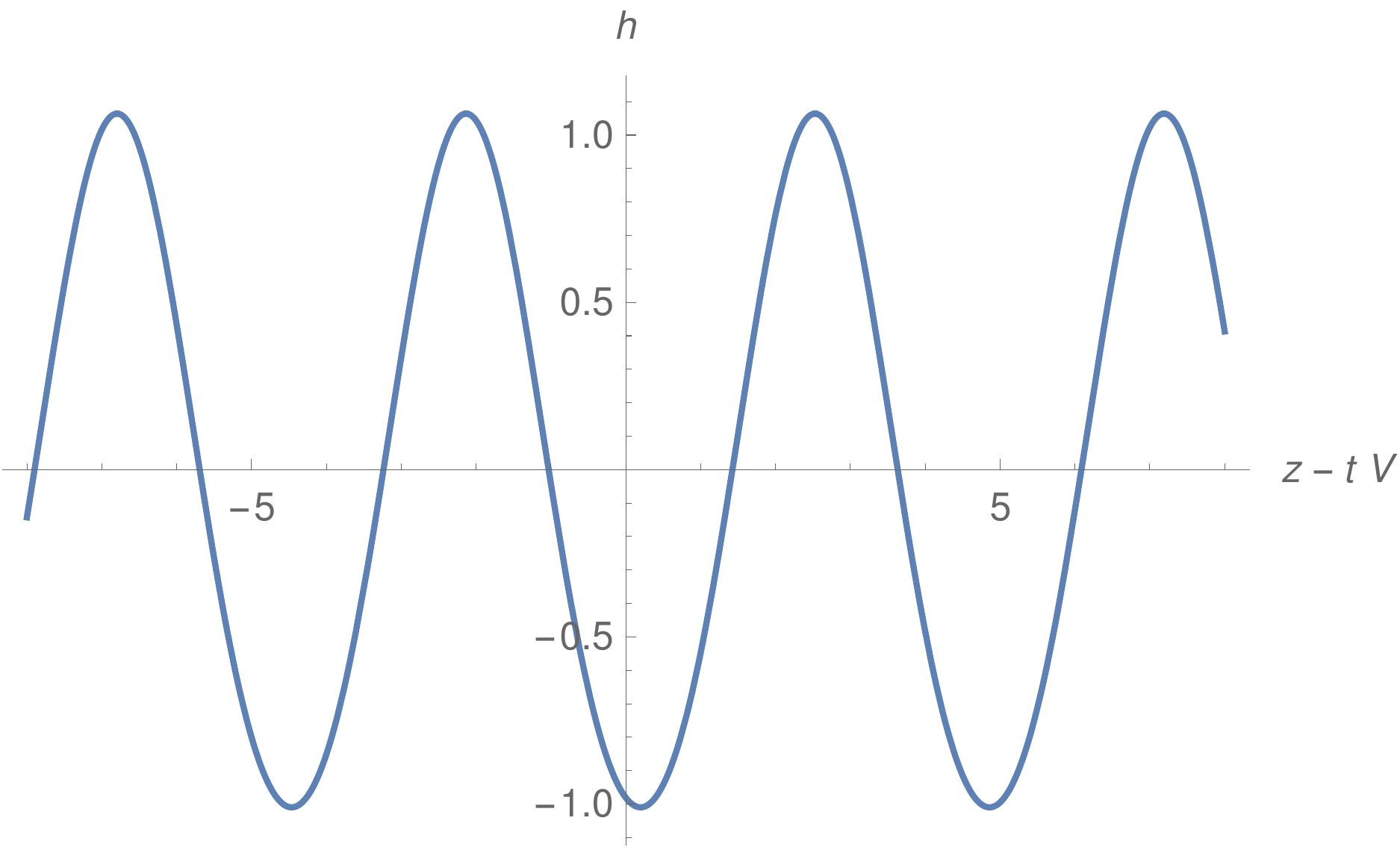}
\caption{Plot of the solution~\eqref{weierstrass-2} corresponding to the three real root case. The parameters have the following values:
$Y=\omega_0=g=1$, $\lambda=V=0.5$, $\chi=2$, $\kappa=3$, $z_0=0$.}
\label{w2}
\end{figure}

\subsubsection{ELEMENTARY SOLUTIONS} These solutions correspond to the cases of degenerate roots and could be directly deduced as particular cases of the Weierstrass cases. However, since classifying all possible Weierstrass configurations is quite cumbersome,
as shown in Appendix \ref{appendice-B}, it is easier to construct them directly.
In order to obtain a solution expressible in an algebraic form, we set equal to zero the discriminant of the fourth degree polynomial at the denominator of~\eqref{quartic},
and resolve it for the constant $\lambda$ (here $\Omega:=\omega_0^2-g^2v^2$):
\begin{equation}
Y^2\lambda_\pm =
-\frac{3}{64\chi^3}
\Bigl( 27g^4\kappa^4V^4+72g^2\kappa^2V^2\chi\Omega+32\chi^2\Omega^2 \pm
\sqrt{g^2\kappa^2V^2(9g^2\kappa^2V^2+16\chi\Omega)^3} \Bigr) ;
\end{equation}
$\lambda_\pm$ is real if and only if $\Omega \chi \geq -9g^2\kappa^2V^2/16$. Recall that at least one between $\lambda_+$ and $\lambda_-$ must be positive in order to be consistent with our physical assumptions.

On the other hand, with the aid of the general theory of quartic equations \cite{Dickson},
we discover the nature of the roots of our polynomial, depending of how the parameters change.
It holds (one must choose the positive one between $\lambda_+$ and $\lambda_-$):
\begin{itemize}
\item if $\Omega<0$ and
$-\frac{3\Omega^2}{2\lambda_\pm} < \chi < \frac{\Omega^2}{2\lambda_\pm}$,
then we have four real roots, of which two are double;
\item if $\Omega<0$ and $\chi<-\frac{3\Omega^2}{2\lambda_\pm}$, there are a double real root and two complex conjugate roots;
\item if $\Omega>0$ and $\chi<0$, then there are again a double real root and two complex conjugate roots.
\end{itemize}
This is what the theory tells us. We can say something more: from the equation $\frac{\lambda}{12}h^4+\Omega h^2+2g\kappa V h-2\chi=0$, we see that if $\chi>0$, it can be chosen $h$ such that the polynomial is zero.
So, we conclude that for every $\chi>0$ the equation has two double real roots.

\begin{figure}
\centering
\includegraphics[scale=0.5]{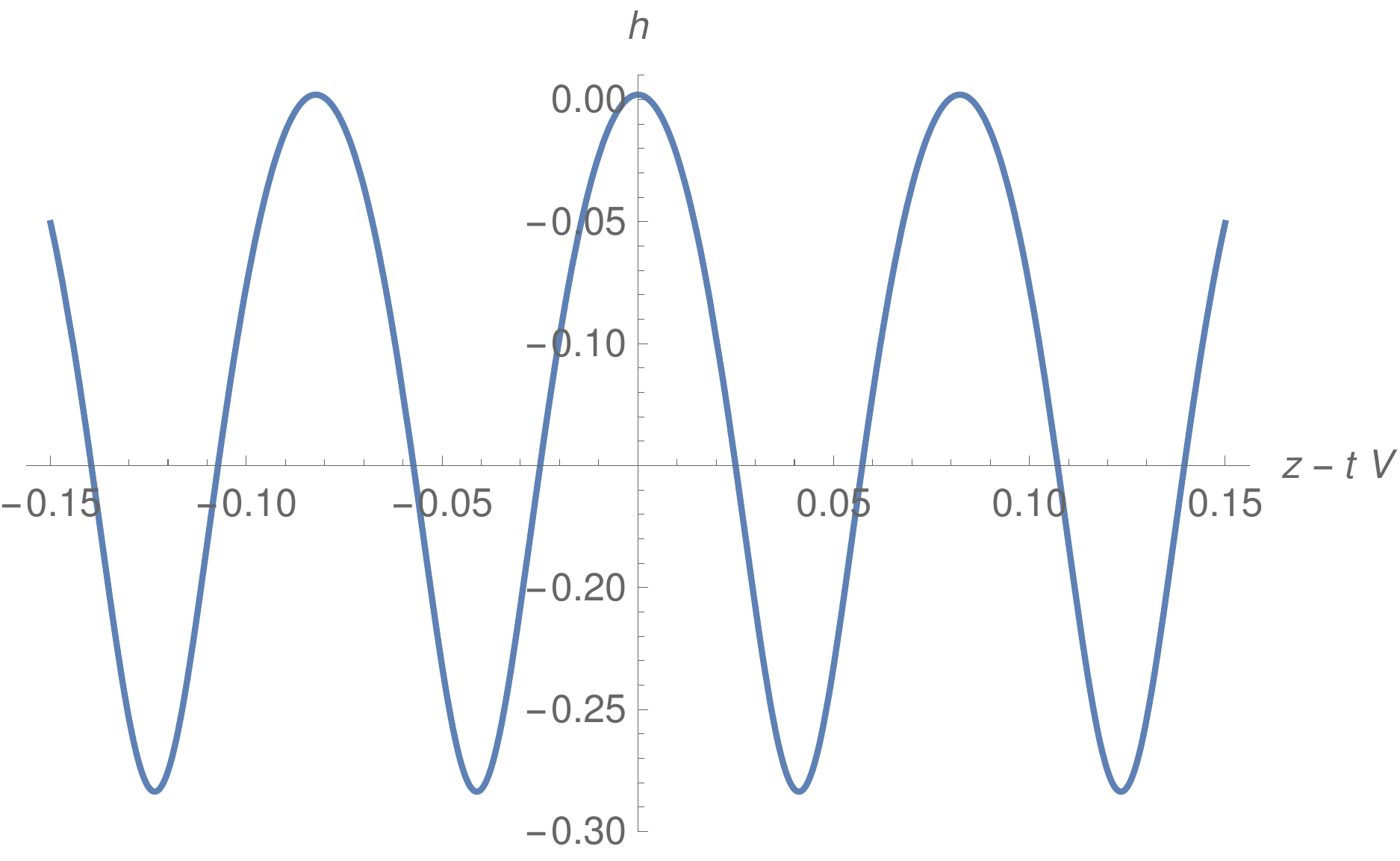}
\caption{Plot of the solution~\eqref{weierstrass-1} corresponding to the one real root case. The parameters have the following values:
$Y=10^3$, $\omega_0=g=1$, $\lambda=V=0.5$, $\chi=2$, $\kappa=10^3$, $z_0=0$.}
\label{w3}
\end{figure}

Note that not all of these conditions are compatible with our choice of $\lambda_\pm$.
Consider, for example, $\lambda_+$: if we want to guarantee the positivity of $\lambda_+$, then the combination $\Omega>0$, $\chi>0$ is not acceptable, while the other combinations allow $\lambda_+>0$. Moreover, the reality condition 
$\chi\Omega \geq -9g^2\kappa^2V^2/16$ provide further constraints.
Independently by the value of the parameters, the general solution of~\eqref{quartic} can be easily written: if $\alpha$ is the double real root, and $\beta$, $\delta$ are real (or complex), then, recalling that in our case $\delta+\beta=-2\alpha$,
\begin{equation}
h =
\Biggl[
\frac{\beta-\delta}{2(\alpha-\beta)(\alpha-\delta)}
\cos\biggl(\sqrt{\frac{\lambda(\alpha-\beta)(\alpha-\delta)}{12}} \frac{z-Vt}{V}\biggr)
- \frac{2\alpha}{(\alpha-\beta)(\alpha-\delta)} \Biggr]^{-1} + \alpha .
\end{equation}
The term $(\beta-\delta)$ shows that if $\beta$ and $\delta$ are not real, then the final solution is not real too and must be excluded. So, the parameters must be such that the roots are all real. In this case
\begin{align}
 (\alpha-\beta)(\alpha-\delta)=4\alpha^2-(\alpha+\beta)^2=4\alpha^2-(\alpha+\delta)^2
\end{align}
shows that we may have both signs for $(\alpha-\beta)(\alpha-\delta)$, so that we have
\begin{align}
h =
\Biggl[
\frac{\beta-\delta}{2(\alpha-\beta)(\alpha-\delta)}
\cos\biggl(\sqrt{\frac{\lambda(\alpha-\beta)(\alpha-\delta)}{12}} \frac{z-Vt}{V}\biggr)
- \frac{2\alpha}{(\alpha-\beta)(\alpha-\delta)} \Biggr]^{-1} + \alpha, &\ {\rm if}\ (\alpha-\beta)(\alpha-\delta)>0,\\
h =
\Biggl[
\frac{\beta-\delta}{2(\alpha-\beta)(\alpha-\delta)}
\cosh\biggl(\sqrt{\frac{\lambda(\alpha-\beta)(\alpha-\delta)}{12}} \frac{z-Vt}{V}\biggr)
- \frac{2\alpha}{(\beta-\alpha)(\alpha-\delta)} \Biggr]^{-1} + \alpha, &\ {\rm if}\ (\alpha-\beta)(\alpha-\delta)<0,
\end{align}
The last case includes the simplest situation $k=0$, $\chi=0$ studied above. The first case, instead, includes the simple case studied below.

\subsection{$\kappa\neq0$, $\chi=0$}
As a particular subcase, we can study equation~\eqref{quartic} with $\chi=0$;
here the analysis is simplified, because we have to study a cubic equation.
Requiring the discriminant of the cubic to be zero, we find
\begin{equation}
\lambda = -\frac{4\Omega^3}{9g^2\kappa^2V^2} ,
\end{equation}
so the positivity of $\lambda$ requires $\Omega<0$, that is $v>\omega_0/g$ as usual.

With this value for $\lambda$ there are three real roots, of which two are double.
The double root is $-3g\kappa V/\Omega$, and the simple root is $6g\kappa V/\Omega$.
The solution of~\eqref{quartic} is
\begin{equation}
h =
-\frac{3g\kappa V}{\omega_0^2-g^2v^2} \Biggl[
\frac{3}{\cos\Bigl( \frac{3g\kappa\sqrt{\lambda}}{2(\omega_0^2-g^2v^2)}(z-Vt) \Bigr) - 2}
+ 1 \Biggr] .
\end{equation}

\section{Conclusions}

We have studied, in the simplified framework of the $\phi\psi$-model, 
the propagation of the electromagnetic field in a spatially finite sample of dielectric medium. This situation 
is physically relevant in the Analogue  Gravity picture for the Hawking effect, as experiments involve necessarily 
finite samples of dielectrics. We have chosen to work in a cylindrical geometry, where the dielectric field fills only 
a finite cylindrical region of length $2L$ and radius $R$. The remaining region of radius $R$ is 
filled by vacuum. This may be considered as a model for a optic fiber, which are an active benchmark 
for experiments in Analogue Gravity \cite{drori}.

Our present study concerns analytical properties of the solutions for the equations of motion of the involved fields. 
As a preliminary analysis, we have considered the  boundary conditions to 
be imposed on the fields, together with a complete scattering basis and the quantization of the fields in the case of a still homogeneous dielectric sample. We have also described the propagator for the fields in the given setting.

Then we have introduced a nonlinearity in the model, 
as the dielectric media we are interested in must be associated with the Kerr effect. Indeed, as pointed out firstly in \cite{philbin}, 
a possibility to obtain analogous black hole in dielectrics consists in generating strong laser pulses which propagate inside 
a nonlinear dielectric medium. The Kerr effect gives rise to a propagating perturbation of the refractive index which plays the 
role of the analogous black hole, and is indeed involved with a horizon. Our interest has been to find out solitonic solutions 
describing the aforementioned perturbation, i.e. the background solutions around which a linearization is performed, 
and the perturbations are quantized.  
We have shown that solitonic solutions exist, representing a dielectric perturbation travelling 
with constant velocity in the direction of the cylindrical fiber axis.

Further developments can involve different aspects, all with a noticeable physical interest. One may study 
perturbation theory for the model, looking for quantum effects induced by surface effects, e.g.~transition radiation 
\cite{ginzburg}. Also, one can study absorption in the model, associated with the fourth order perturbation. 
Our main research focus, which is represented by the analogous Hawking effect, requires the analysis 
of the linearization of the model around the solitonic solution and its quantization. One may also limit to 
consider, in the comoving frame of the dielectric perturbation, the dependence of the dielectric susceptibility 
and of the resonance frequency on space (induced by the Kerr effect), and analyze the  Hawking effect with 
simpler background profiles. Future works will be devoted to the aforementioned goals.

\section*{Acknowledgements}

A.V.~was partially supported by Ministero dell'Universit\`a e della Ricerca MIUR-PRIN contract 2017CC72MK\_003.

%

\begin{appendix}
\section{Orthogonality relations}\label{orthogonality}
For
\begin{eqnarray}
f=
\begin{pmatrix}
\phi \\ \psi
\end{pmatrix},
\qquad
\tilde f=
\begin{pmatrix}
\tilde \phi \\ \tilde \psi
\end{pmatrix}
,
\end{eqnarray}
we define
\begin{align}
j^\mu_{f,\tilde f}(\pmb x):= i [\phi^*(\pmb x) \partial^\mu \tilde \phi(\pmb x)-\tilde\phi(\pmb x) \partial^\mu \phi^*(\pmb x)+v^\mu(\psi^*(\pmb x)\dot {\tilde \psi}(\pmb x)-\tilde\psi(\pmb x)\dot {\psi}^*(\pmb x))
+gv^\mu (\psi^*(\pmb x) \tilde \phi(\pmb x)-\phi^*(\pmb x)\tilde \psi(\pmb x))].
\end{align}
It is a conserved current
\begin{eqnarray}
\partial_\mu j^\mu_{f,\tilde f}(\pmb x)=0,
\end{eqnarray}
and the scalar product is
\begin{eqnarray}
(f|\tilde f)=\int_{C_t} j^0_{f,\tilde f}(\pmb x),
\end{eqnarray}
where $C_t$ is the slice of $C$ obtained by fixing $t$.

The conservation law is particularly helpful for computing the scalar product among plane wave solutions or scattering wave solutions. These solutions have the form
\begin{eqnarray}
\phi(\pmb x)=e^{-i\omega_{ksm}t} \varphi_{_{ksm}}(\vec x) ,
\end{eqnarray}
so if we take
\begin{eqnarray}
f(\pmb x)=e^{-i\omega_{ksm}t} 
\begin{pmatrix}
\varphi_{ksm}(\vec x) \\
\varrho_{ksm}(\vec x)
\end{pmatrix},
\qquad
\tilde f(\pmb x)=e^{-i\omega_{k's'm'}t} 
\begin{pmatrix}
\varphi_{k's'm'}(\vec x) \\
\varrho_{k's'm'}(\vec x)
\end{pmatrix},
\end{eqnarray}
then
\begin{eqnarray}
\partial_0  j^0_{f,\tilde f}(\pmb x)=i(\omega_{ksm}-\omega_{k's'm'}) j^0_{f,\tilde f}(\pmb x),
\end{eqnarray}
and integrating over the spatial slice and using the continuity equation we get
\begin{align}
(f|\tilde f)=-\frac 1{i(\omega_{ksm}-\omega_{k's'm'})} \int_{\partial C_t} \vec j_{f,\tilde f}(\pmb x) \cdot \vec n(\pmb x) \ d^2\sigma(\pmb x),
\end{align}
where we have used the continuity of $\vec j_{f,\tilde f}(\pmb x) \cdot \vec n(\pmb x)$, as a consequence of the boundary conditions.
In order to compute this integral let us restrict it on the compact cylinder 
\begin{eqnarray}
C_t^Z=\{(\rho,\theta,z)\in C_t| -Z\leq z\leq Z\},
\end{eqnarray}
so that
\begin{align}
(f|\tilde f)=-\frac 1{i(\omega_{ksm}-\omega_{k's'm'})} \lim_{Z\to+\infty}\int_{\partial C^Z_t} \vec j_{f,\tilde f}(\pmb x) \cdot \vec n(\pmb x) \ d^2\sigma(\pmb x).
\end{align}
Finally, by taking into account the boundary condition for the fields we get
\begin{align}
(f|\tilde f)=-\frac 1{i(\omega_{ksm}-\omega_{k's'm'})} \lim_{Z\to+\infty} \left( \int_{D_Z} j^z_{f,\tilde f}(\pmb x)  \rho^2 d\rho d\theta -\int_{D_{-Z}} j^z_{f,\tilde f}(\pmb x)  \rho^2 d\rho d\theta \right),
\end{align}
where 
\begin{eqnarray}
D_{\pm Z}=\{(\rho,\theta,z)\in C_t| z=\pm Z \}.
\end{eqnarray}
Using
\begin{eqnarray}
\int_0^{2\pi} e^{i(m'-m)\theta} d\theta =2\pi \delta_{mm'} ,
\end{eqnarray}
for two right dielectric scattering functions we get
\beq
\begin{split}
\frac {(f|\tilde f)}{\kappa^*_{ksm}\kappa_{ksm}}&=-\frac {e^{i(\omega_{ksm}-\omega_{k's'm})t}}{i(\omega_{ksm}-\omega_{k's'm})} 2\pi \delta_{mm'}
\int_0^R J_{|m|} \biggl(z_{ms}\frac \rho R\biggr) J_{|m|} \biggl(z_{ms'}\frac \rho R\biggr) \rho d\rho\ \\
&\quad \times \lim_{Z\to+\infty}\Bigl[ (k-k') R^*_{ksm} e^{-i(k+k')Z} -(k-k') R_{k's'm} e^{i(k+k')Z} +(k+k') \\
&\quad \times \Bigl( (T^*_{ksm} T_{k's'm}+R^*_{ksm} R_{k's'm})e^{-(k-k')Z} -e^{i(k-k')Z}\Bigl) \Bigr].
\end{split}
\eeq
We first show that 
\begin{lemma}
It holds
\begin{align}
\int_0^R J_{|m|} \biggl(z_{ms}\frac \rho R\biggr) J_{|m|} \biggl(z_{ms'}\frac \rho R\biggr) \rho d\rho=\delta_{ss'} \frac {R^2}{2z^2_{ms}}(z_{ms}^2-m^2)J^2_m(z_{ms}).
\end{align}
\end{lemma}
\begin{proof}
The Bessel equation can be written in the form
\begin{eqnarray}
\label{Bessel}
\frac {d}{d\rho} \biggl( \rho \frac d{d\rho} J_{|m|} \biggl(z_{ms}\frac \rho{R}\biggl) \biggr) +\biggl(\rho \frac {z^2_{ms}}R^2-\frac {m^2}\rho \biggr)J_{|m|}\biggl(z_{ms}\frac \rho{R}\biggr)=0 ,
\end{eqnarray}
from which we get
\beq
\begin{split}
& \frac d{d\rho} \biggl[ \rho J_{|m|}\biggl(z_{ms'}\frac \rho{R}\biggr) \frac d{d\rho} J_{|m|}\biggl(z_{ms}\frac \rho{R}\biggr)-\rho J_{|m|}\biggl(z_{ms}\frac \rho{R}\biggr) \frac d{d\rho} J_{|m|}\biggl(z_{ms'}\frac \rho{R}\biggr) \biggr] \\
&\quad +\frac {z^2_{ms}-z^2_{ms'}}{R^2} \rho J_{|m|} \biggl(z_{ms}\frac \rho{R}\biggl) J_{|m|} \biggl(z_{ms'}\frac \rho{R}\biggr)=0 ,
\end{split}
\eeq
that integrated from $0$ to $R$ in $d\rho$ gives
\begin{eqnarray}
\int_0^R J_{|m|} \biggl(z_{ms}\frac \rho R\biggl) J_{|m|} \biggl(z_{ms'}\frac \rho R\biggr) \rho d\rho=0 ,
\end{eqnarray}
if $s\neq s'$. Moreover, from~\eqref{Bessel} we get
\beq
\begin{split}
0&=\frac d{d\rho} \left( \rho J_{|m|} (z_{ms}\frac \rho{R}) \frac d{d\rho} \left( \rho \frac d{d\rho} J_{|m|} (z_{ms}\frac \rho R)   \right) \right)+\frac {z^2_{ms}}{R^2} 2\rho J^2_{m}(z_{ms}\frac \rho{R}) \\
&\quad +(\rho^2\frac {z^2_{ms}}{R^2}-m^2) 2J_{|m|}(z_{ms}\frac \rho{R}) \frac d{d\rho}J_{|m|}(z_{ms}\frac \rho{R})\\
&=\frac d{d\rho} \left( \rho J_{|m|} (z_{ms}\frac \rho{R}) \frac d{d\rho} \left( \rho \frac d{d\rho} J_{|m|} (z_{ms'}\frac \rho R)   \right) \right)+\frac {z^2_{ms}}{R^2} 2\rho J^2_{m}(z_{ms}\frac \rho{R}) \\
&\quad -\frac d{d\rho} \left( \rho  \frac d{d\rho}J_{|m|}(z_{ms}\frac \rho{R})  \right)^2 .
\end{split}
\eeq
After integration in $d\rho$ from $0$ to $R$, and using the definition of $z_{ms}$, we get
\begin{eqnarray}
2\frac {z^2_{ms}}{R^2} \int_0^R J^2_{|m|} \biggl(z_{ms}\frac \rho R\biggr) \rho d\rho=-R^2 J^2_{|m|} (z_{ms}) \frac {d^2}{d\rho^2} J^2_{|m|} \biggl(z_{ms}\frac \rho R\biggr)\bigg|_{\rho=R}.
\end{eqnarray}
Using again the Bessel equation we finally get the assert. 
\end{proof}
Thus
\beq
\begin{split}
\frac {(f|\tilde f)}{\kappa^*_{ksm}\kappa_{ksm}}=&-\frac {e^{i(\omega_{ksm}-\omega_{k'sm})t}}{i(\omega_{ksm}-\omega_{k'sm})} 2\pi \delta_{mm'} \delta_{ss'} \frac {R^2}{2z^2_{ms}}(z_{ms}^2-m^2)J^2_m(z_{ms})
\ \lim_{Z\to+\infty}\left[ (k-k') R^*_{ksm} e^{-i(k+k')Z} \right. \cr
& \left. -(k-k') R_{k'sm} e^{i(k+k')Z} +(k+k') \left( (T^*_{ksm} T_{k'sm}+R^*_{ksm} R_{k'sm})e^{-(k-k')Z} -e^{i(k-k')Z}\right) \right].
\end{split}
\eeq
In order to compute this limit, we rewrite it in the form
\beq
\begin{split}
\lim_{Z\to\infty} \frac {k-k'}{\omega_{ksm}-\omega_{k'sm}} &\left[ (k+k') \frac {e^{-(k-k')Z} -e^{i(k-k')Z}}{k'-k} +\left( R^*_{ksm} e^{-i(k+k')Z} -R_{k'sm} e^{i(k+k')Z} \right) \right. \cr
&\left. +(k+k') \frac {T^*_{ksm} T_{k'sm}+R^*_{ksm} R_{k'sm}-1}{k'-k}e^{-(k-k')Z} \right].
\end{split}
\eeq
Since $k$ and $k'$ are positive, the second term in the square brackets vanishes in the limit because of the Riemann--Lebesgue theorem.  
The second term vanishes for the same reason unless $k=k'$. Since $|T_{ksm}|^2+|R_{ksm}|^2=1$, it stays finite for $k=k'$ (if we take the continuation by the limit $k'\to k$). Thus, in the distributional sense, it vanishes.
So, the surviving limit is
\begin{align}
\lim_{Z\to\infty} \frac {k-k'}{\omega_{ksm}-\omega_{k'sm}} &\left[ (k+k') \frac {2i \sin ((k-k')Z)}{k'-k} \right]= 2\pi k \frac {dk}{d\omega} 2i \delta(k'-k) ,  \label{limit}
\end{align}
and we finally get
\begin{align}
(f|\tilde f)=|\kappa_{ksm}|^2 4\omega_{ksm}\pi^2 R^2 \left( 1-\frac {m^2}{z^2_{ms}} \right) J^2_m(z_{ms}) \delta_{mm'} \delta_{ss'} \delta(k'-k).
\end{align}
The same result is true for two left dielectric wave functions.

The same procedure can be used to compute the scalar product between two right (or left) gap wave functions $g$, $\tilde g$
\begin{align}
(g|\tilde g)=|\tilde \kappa_{ksm}|^2 \omega_{ksm} \pi^2 R^2 \left( 1-\frac {m^2}{z^2_{ms}} \right) J^2_m(z_{ms}) \delta_{mm'} \delta_{ss'} \delta(k'-k).
\end{align}
All other combinations vanish. It is worth to mention that for the particular case $m=0$ there is also the zero $z_{0,0}=0$, for which
\begin{align}
 (f|\tilde f)&=|\kappa_{k00}|^2 4 k\pi^2 R^2 \delta(k'-k), \\
 (g|\tilde g)&=|\tilde \kappa_{k00}|^2 k \pi^2 R^2 \delta(k'-k). 
\end{align}

\section{Study of equation~\eqref{quartic}}
\label{appendice-B}
Let us write the quartic as
\begin{align}
 \frac{Y^2\lambda}{12}h^4 + (\omega_0^2-g^2V^2\gamma^2)h^2 + 2g\kappa Vh - 2\chi=p(h) ,
\end{align}
with
\begin{align}
 p(x)=\alpha_0 x^4+\alpha_2 x^2+\alpha_3 x+\alpha_4 \equiv \alpha_0 (x-E_0)(x-E_1)(x-E_2)(x-E_3),
\end{align}
where $E_j$ are the polynomial roots, satisfying $E_0+E_1+E_2+E_3=0$. We are assuming that there is at least one real root and so define $E_0$ to be the largest real root. Let us consider the change of variables
\begin{align}
 h=\frac {as+b}{cs+d},
\end{align}
with $a$, $b$, $c$, $d$ all real. These parameters are defined up to a global real rescaling, which can be fixed so that $ad-bc=\varepsilon$, with $\varepsilon=\pm1$. We get
\begin{align}
 h'=\frac {\varepsilon s'}{(cs+d)^2}, 
\end{align}
so that
\begin{align}
 \frac {h'}{\sqrt {p(h)}}=\frac {\varepsilon s'}{q(s)},
\end{align}
where
\begin{align}
 q(s)=a_0 s^4+a_1 s^3+a_2 s^2+a_3 s+ a_4,
\end{align}
with
\beq
 a_0 = c^4 p\biggl(\frac ac\biggr),
\eeq
which we easily set to zero by imposing
\begin{align}
 a=c E_0.
\end{align}
With this position, the remaining coefficients are 
\begin{align}
 a_1&=c^3 (b-E_0 d)\ p'(E_0), \\ 
 a_2&=(b-E_0 d)c^2 d\ p'(E_0) -(b-E_0 d) \varepsilon \frac c2\ p''(E_0), \\
 a_3&=4cd^3\ p\biggl(\frac bd\biggr)+\varepsilon d^2\ p'\biggl(\frac bd\biggr),\\
 a_4&=d^4\ p\biggl(\frac bd\biggr).
\end{align}
We have to impose the condition $a_2=0$ and $a_1=4$. Notice that $b-E_0 d=-\varepsilon/c$. Moreover, since we are assuming the roots are generic, therefore all distinct, $p'(E_0)\neq 0$. We finally get:
\beq
a=E_0 c, \qquad b = \frac c4 \left(p'(E_0)-\frac 12 E_0 p''(E_0)\right), \qquad
c= \frac 2{\sqrt {|p'(E_0)|}} \qquad d=-\frac {p''(E_0)}8 c .
\eeq
and $\varepsilon= -{\rm sign}(p'(E_0))$.
With the assumption $E_0$ real, these coefficients are all real and lead us to equation \eqref{equation}, with 
\beq
 g_2=4cd^3\ p\biggl(\frac bd\biggr)+\varepsilon d^2\ p'\biggl(\frac bd\biggr), \qquad\ g_3=d^4\ p\biggl(\frac bd\biggr).
\eeq
The conditions leading to this solution are essentially the ones guaranteeing the existence of at least one real solution of $p(x)=0$.

The discriminant of our quartic equation is
\beq
 \Delta=256 \alpha_0^3\alpha_4^3-128 \alpha_0^2 \alpha_2^2 \alpha_4^2+144 \alpha_0^2\alpha_2 \alpha_3^2 \alpha_4 -27 \alpha_0^2 \alpha_3^4+16 \alpha_0 \alpha_2^4\alpha_4-4\alpha_0 \alpha_2^3 \alpha_3^2.
\eeq
From~\cite{Dickson}, Theorem 7, we see that
if $\Delta<0$ there are always 2 real roots and 2 complex conjugate roots.
For $\Delta=0$ we boil down to the case of degenerate solutions, studied apart in the main text.
Finally, if $\Delta>0$ the only case with real roots is when
the conditions $M\equiv\alpha_0 \alpha_2<0$ and $N\equiv4\alpha_0 \alpha_4-\alpha_2^2<0$ are satisfied.
If the case, then there are four real solutions.

More explicitly, since $\alpha_0=\frac \lambda{12} Y^2>0$, we have to consider the sign of
\beq
 \tilde \Delta\equiv\frac {\Delta}{16\alpha_0}=-128 \alpha_0^2 \chi^3-32 \alpha_0\Omega^2 \chi^2-72\alpha_0 \Omega g^2\kappa^2 V^2 \chi-27\alpha_0 g^4 \kappa^4 V^4-2\Omega^4 \chi-g^2\kappa^2V^2 \Omega^3,
\eeq
where
\begin{align}
 \Omega=\omega_0^2-g^2 V^2\gamma^2.
\end{align}
Moreover,
\begin{align}
\frac M{\alpha_0}&\equiv \alpha_2= \Omega, \\
N&\equiv 4\alpha_0\alpha_4-\alpha_2^2=-8\alpha_0 \chi-\Omega^2.
\end{align}
Therefore, if $\chi\geq0$ and $\Omega>0$, $\tilde\Delta<0$ we have always two real roots and two complex conjugate roots. If $\chi\geq 0$ and $\omega<0$ then $\tilde\Delta$ may have any sign but $M$ and $N$ are both negative, so we always have
two or four real roots. This was also evident from the fact that $p(0)=-2\chi$ so, if $\chi>0$, we always have at least one real root (two if $\chi>0$ or if $\kappa\neq0$ when $\chi=0$. For $\chi=\kappa=0$, $x=0$ is a double root of $p$ and the discriminant 
vanishes).

When $\chi$ is negative things are little bit more complicate. In this case one has to study more carefully the sign of $\tilde \Delta$. When it is negative then we are done, while when it is positive then $\Omega$ must be negative, providing the same condition 
$v>\omega_0/g$ as for the case $\chi=\kappa=0$. In this case, we have also to impose the condition $N<0$, which gives $g^2 v^2>\omega_0^2+\sqrt {8\alpha_0 |\chi|}$, that is
\beq
 V^2> \frac {\omega_0^2+|Y|\sqrt {\frac 23 \lambda |\chi|}}{g^2+\omega_0^2+|Y|\sqrt {\frac 23 \lambda |\chi|}},
\eeq
which generalizes condition (\ref{condition}).

So, one is left with the study of the general conditions for which $\tilde \Delta$ have a specific sign when $\chi$ is negative. We will not pursue this here, but we limit ourselves to the following considerations. We can look at $\tilde \Delta>0$ as a second order 
inequality in $\alpha_0$, recalling the physical constraint $\alpha_0>0$. Since $\chi$ is negative, this is always true if the discriminant is negative, while if it is positive we are led to the condition $\alpha_0> \max (0, \alpha_+)$, $\alpha_+$ being
the higher root of the quadric. A complete classification of all possibilities is not difficult but quite cumbersome and out of the task of the present work.

\end{appendix}



\begin{thebibliography}{99}

\bibitem{philbin}
T.~G. Philbin, C.~Kuklewicz, S.~Robertson, S.~Hill, F.~K{\"o}nig, and
  U.~Leonhardt, ``Fiber-optical analog of the event horizon,'' {\em Science},
  vol.~319, no.~5868, pp.~1367--1370, 2008.

\bibitem{cacciatori} 
  S.~L.~Cacciatori, F.~Belgiorno, V.~Gorini, G.~Ortenzi, L.~Rizzi, V.~G.~Sala and D.~Faccio,
  ``Space-time geometries and light trapping in travelling refractive index perturbations,''
  New J.\ Phys.\  {\bf 12}, 095021 (2010).

\bibitem{belgiorno-prd} 
  F.~Belgiorno, S.~L.~Cacciatori, G.~Ortenzi, L.~Rizzi, V.~Gorini and D.~Faccio,
  ``Dielectric black holes induced by a refractive index perturbation and the Hawking effect,''
  Phys.\ Rev.\ D {\bf 83}, 024015 (2011)
  doi:10.1103/PhysRevD.83.024015
  [arXiv:1003.4150 [quant-ph]].

\bibitem{rubino2011}
E.~Rubino, F.~Belgiorno, S.~L. Cacciatori, M.~Clerici, V.~Gorini, G.~Ortenzi,
  L.~Rizzi, V.~G. Sala, M.~Kolesik, and D.~Faccio, ``{Experimental evidence of
  analogue Hawking radiation from ultrashort laser pulse filaments},'' {\em New
  J. Phys.}, vol.~13, p.~085005, 2011.

\bibitem{petev}
M.~Petev, N.~Westerberg, D.~Moss, E.~Rubino, C.~Rimoldi, S.~L. Cacciatori,
  F.~Belgiorno, and D.~Faccio, ``{Blackbody emission from light interacting
  with an effective moving dispersive medium},'' {\em Phys. Rev. Lett.},
  vol.~111, p.~043902, 2013.

\bibitem{finazzi2012}
S.~Finazzi and I.~Carusotto, ``{Quantum vacuum emission in a nonlinear optical
  medium illuminated by a strong laser pulse},'' {\em Phys. Rev.}, vol.~A87,
  no.~2, p.~023803, 2013.


\bibitem{finazzi2013}
S.~Finazzi and I.~Carusotto, ``{Spontaneous quantum emission from analog white
  holes in a nonlinear optical medium},'' {\em Phys. Rev.}, vol.~A89, no.~5,
  p.~053807, 2014.

\bibitem{hopfield-hawking} 
  F.~Belgiorno, S.~L.~Cacciatori and F.~Dalla Piazza,
  ``Hawking effect in dielectric media and the Hopfield model,''
  Phys.\ Rev.\ D {\bf 91}, no. 12, 124063 (2015)
  doi:10.1103/PhysRevD.91.124063
  [arXiv:1411.7870 [gr-qc]].

\bibitem{jacquet}
M.~Jacquet and F.~K\"onig, ``Quantum vacuum emission from a refractive-index
  front,'' {\em Phys. Rev. A}, vol.~92, p.~023851, Aug 2015.

\bibitem{linder} 
  M.~F.~Linder, R.~Schutzhold and W.~G.~Unruh,
  ``Derivation of Hawking radiation in dispersive dielectric media,''
  Phys.\ Rev.\ D {\bf 93}, no. 10, 104010 (2016)
  doi:10.1103/PhysRevD.93.104010
  [arXiv:1511.03900 [gr-qc]].

\bibitem{hopfield-kerr} 
  F.~Belgiorno, S.~L.~Cacciatori, F.~Dalla Piazza and M.~Doronzo,
  ``Hopfield-Kerr model and analogue black hole radiation in dielectrics,''
  Phys.\ Rev.\ D {\bf 96}, no. 9, 096024 (2017)
  doi:10.1103/PhysRevD.96.096024
  [arXiv:1707.01663 [hep-th]].

\bibitem{haw-book} 
  F.~D.~Belgiorno, S.~L.~Cacciatori and D.~Faccio,
  {\emph{Hawking Radiation : From Astrophysical Black Holes to Analogous Systems in Lab}}; 
  World Scientific Publishing Company, Singapore (2018).


\bibitem{master}
F.~Belgiorno, S.~L.~Cacciatori and A.~Vigan\`o,
  ``Analog Hawking effect: A master equation,''
  Phys.\ Rev.\ D {\bf 102}, no. 10, 105003 (2020)
  doi:10.1103/PhysRevD.102.105003
  [arXiv:2003.04236 [gr-qc]].


\bibitem{faccio}
F.~Belgiorno, S.~L. Cacciatori, M.~Clerici, V.~Gorini, G.~Ortenzi, L.~Rizzi,
  E.~Rubino, V.~G. Sala, and D.~Faccio, ``{Hawking radiation from ultrashort
  laser pulse filaments},'' {\em Phys. Rev. Lett.}, vol.~105, p.~203901, 2010.

\bibitem{unshu} 
  R.~Schutzhold and W.~G.~Unruh,
  ``Comment on: Hawking Radiation from Ultrashort Laser Pulse Filaments,''
  Phys.\ Rev.\ Lett.\  {\bf 107}, 149401 (2011)
  doi:10.1103/PhysRevLett.107.149401
  [arXiv:1012.2686 [quant-ph]].


\bibitem{rebut} 
  F.~Belgiorno {\it et al.},
  ``Reply to Comment on: Hawking radiation from ultrashort laser pulse filaments,''
  Phys.\ Rev.\ Lett.\  {\bf 107}, 149402 (2011)
  doi:10.1103/PhysRevLett.107.149402
  [arXiv:1012.5062 [quant-ph]].


\bibitem{drori}
J.~Drori, Y.~Rosenberg, D.~Bermudez, Y.~Silberberg and U.~Leonhardt,
  ``Observation of Stimulated Hawking Radiation in an Optical Analogue,''
  Phys.\ Rev.\ Lett.\  {\bf 122}, no. 1, 010404 (2019)
  doi:10.1103/PhysRevLett.122.010404
  [arXiv:1808.09244 [gr-qc]].

\bibitem{phi-psi} 
  F.~Belgiorno, S.~L.~Cacciatori, F.~Dalla Piazza and M.~Doronzo,
  ``${\Phi -\Psi }$ model for electrodynamics in dielectric media: exact quantisation in the Heisenberg representation,''
  Eur.\ Phys.\ J.\ C {\bf 76}, no. 6, 308 (2016)
  doi:10.1140/epjc/s10052-016-4146-1
  [arXiv:1512.08738 [math-ph]].








\bibitem{hopfield-scripta}
F.~Belgiorno, S.~L. Cacciatori, and F.~Dalla~Piazza, ``{The Hopfield model
  revisited: Covariance and Quantization},'' {\em Phys. Scripta}, vol.~91,
  no.~1, p.~015001, 2016.


\bibitem{exact}
F.~Belgiorno, S.~L. Cacciatori, F.~Dalla~Piazza, and M.~Doronzo, ``{Exact
  quantisation of the relativistic Hopfield model},'' {\em Annals Phys.},
  vol.~374, pp.~338--365, 2016.

\bibitem{vigano} 
  F.~Belgiorno, S.~L.~Cacciatori and A.~Vigan\`o,
  ''Spectral boundary conditions and solitonic solutions in a classical Sellmeier dielectric,''
  Eur.\ Phys.\ J.\ C {\bf 77}, no. 6, 404 (2017)
  [arXiv:1610.06308 [hep-th]].



\bibitem{drummond-book}
Drummond, P.D.; Hillery, M.
{\emph{The Quantum Theory of Nonlinear Optics}};
Cambridge University Press: {Cambridge,}  
 UK, 2014.


\bibitem{boyd}
Boyd, R. {\em Nonlinear Optics};
\newblock Academic Press: {Cambridge,} MA, USA, 2008.

\bibitem{ablowitz}
Ablowitz, M.J.
{\emph{Nonlinear Dispersive Waves: Asymptotic Analysis and Solitons}};
Cambridge Texts in Applied~Mathematics;
Cambridge University Press: {Cambridge,} UK, 2011; Volume 47.

\bibitem{universe} 
  F.~Belgiorno and S.~L.~Cacciatori,
  ``Analogous Hawking Effect in Dielectric Media and Solitonic Solutions,''
  Universe {\bf 6}, no. 8, 127 (2020).
  doi:10.3390/universe6080127


\bibitem{ginzburg}
Ginzburg, V.L.
{\emph{Applications of Electrodynamics in Theoretical Physics and Astrophysics}};
Gordon and Breach Science Publishers: {New York,} USA, 1989.



  
 \bibitem{Dickson} 
 L.~E.~Dickson, {\emph{Elementary Theory of Equations}}; John Wiley and Sons; New York 1914.
   
\end{thebibliography}
\end{document}